\newif\ifdraft \draftfalse
\newif\iffull \fulltrue
\newif\ifec \ecfalse

\makeatletter \@input{tex.flags} \makeatother

\documentclass[11pt]{article}

\usepackage{framed}
\usepackage{kpfonts}
\usepackage{url}

\usepackage[suppress]{color-edits}
\addauthor{vs}{red}
\addauthor{ym}{blue}
\addauthor{sw}{red}
\addauthor{as}{purple}

\usepackage[round]{natbib}

\usepackage{algorithm}
\usepackage[noend]{algorithmic}
\usepackage{fullpage}
\usepackage{amsmath, amssymb, amsthm}
\usepackage{verbatim}
\usepackage{color}
\usepackage{xcolor}
\usepackage{multicol}

\definecolor{DarkGreen}{rgb}{0.1,0.5,0.1}
\definecolor{DarkRed}{rgb}{0.5,0.1,0.1}
\definecolor{DarkBlue}{rgb}{0.1,0.1,0.5}
\usepackage[]{hyperref}
\hypersetup{
    unicode=false,          
    pdftoolbar=true,        
    pdfmenubar=true,        
    pdffitwindow=false,      
    pdftitle={},    
    pdfauthor={}
    pdfsubject={},   
    pdfnewwindow=true,      
    pdfkeywords={keywords}, 
    colorlinks=true,       
    linkcolor=DarkRed,          
    citecolor=DarkGreen,        
    filecolor=DarkRed,      
    urlcolor=DarkBlue,          
}
\usepackage{xspace}
\usepackage{cleveref}

\newcommand{\sw}[1]{\ifdraft \textcolor{blue}{[Steven: #1]}\fi}

\newcommand\NN{\mathbb{N}}
\newcommand\RR{\mathbb{R}}
\newcommand\cA{\mathcal{A}}

\newcommand\cC{\mathcal{C}}
\newcommand\cE{\mathcal{E}}
\newcommand\cF{\mathcal{F}}
\newcommand\cG{\mathcal{G}}

\newcommand\cR{\mathcal{R}}
\newcommand\cX{\mathcal{X}}
\newcommand\cD{\mathcal{D}}
\newcommand\cU{\mathcal{U}}

\newcommand{\cI}{\mathcal{I}}

\newcommand{\optdet}{\ensuremath{\mathtt{OPT}^{\mathtt{det}}}}
\newcommand{\optave}{\ensuremath{\mathtt{OPT}^{\mathtt{ave}}}}

\newcommand{\prior}{\psi} 
\newcommand{\joint}{\psi_*} 
\newcommand{\support}{\mathtt{support}}
\newcommand{\trueState}{\theta_0} 

\newcommand{\PiBIC}[1][\mathtt{BIC}]{\Pi_{#1}}
\newcommand{\PiBICdet}[1][\mathtt{BIC}]{\Pi_{#1}^{\mathtt{det}}}
\newcommand{\explorableD}[1][\trueState]{\cA^{\mathtt{det}}_{#1}} 
\newcommand{\pimax}{\pi^{\max}} 

\newcommand{\maxSupport}{{\sf ComputeMaxSupport}}
\newcommand{\maxSupportN}{{\sf ComputeMaxSupport}^\delta}
\newcommand{\MaxEx}{{\sf MaxExplore}} 

\newcommand{\SMaxEx}{{\sf MaxExplore}^{\delta}}
\newcommand{\IndMax}{{\sf RepeatMaxExplore}}  
\newcommand{\SIndMax}{{\sf RepeatMaxExplore}^\delta}

\newcommand{\signal}{\mathtt{AllInfo}}

\newcommand{\signalD}{\mathtt{AllInfo}}
\newcommand{\sigStruc}{S^{\mathtt{str}}} 

\newcommand{\pmin}{p_{\mathtt{min}}} 
\newcommand{\psig}{p_{\mathtt{sig}}} 
\newcommand{\unifLB}{\Lambda^{\delta}_{\mathtt{sig}}} 
\newcommand{\myGap}{\zeta_{\mathtt{sep}}} 

\newcommand{\refeq}[1]{Eq.~(\ref{#1})}
\newcommand{\indicator}[1]{{\bf 1}_{\{#1\}}}
\newcommand{\REW}{\mathtt{REW}}  
\newcommand{\opt}{\ensuremath{\mathtt{OPT}}}
\newcommand{\eps}{\varepsilon}
\def\epsilon{\varepsilon}
\DeclareMathOperator{\OPT}{OPT}

\newcommand{\rbr}[1]{\left(\,#1\,\right)}
\newcommand{\sbr}[1]{\left[\,#1\,\right]}
\newcommand{\cbr}[1]{\left\{\,#1\,\right\}}
\newcommand{\cel}[1]{{\lceil {#1} \rceil}}

\newcommand{\ie}{{\em i.e.,~\xspace}}

\newcommand{\eg}{{\em e.g.,~\xspace}}

\newcommand{\myTab}{\hspace{5mm}}

\DeclareMathOperator{\poly}{poly}

\DeclareMathOperator*{\Expectation}{\mathbb{E}}

\newcommand{\Ex}{{\sf EX}}

\newcommand{\SEx}{{\sf EX}^\delta}

\newcommand{\bic}{{\sf BIC}}

\newcommand{\denoise}{{\sf DeNoise}}

\DeclareMathOperator*{\argmin}{\mathrm{argmin}}

\newtheorem{theorem}{Theorem}[section]
\newtheorem{corollary}[theorem]{Corollary}
\newtheorem{claim}[theorem]{Claim}
\newtheorem{lemma}[theorem]{Lemma}
\newtheorem{definition}[theorem]{Definition}
\newtheorem{assumption}[theorem]{Assumption}
\newtheorem{remark}[theorem]{Remark}

\newcommand{\xhdr}[1]{\vspace{1mm} \noindent{\bf #1}}
\newcommand{\LDOTS}{\, ,\ \ldots\ ,}     

\newcommand{\OMIT}[1]{}

\newenvironment{OneLiners}[1][\ensuremath{\bullet}]
    {\begin{list}
        {#1}
     	{\setlength{\itemsep}{0pt}
	    \setlength{\parsep }{0pt}
      	\setlength{\topsep }{0pt}}}
    {\end{list}}
\newcommand{\E}{\operatornamewithlimits{\mathbb{E}}} 

\newcommand{\RefPropShow}[1]{(P#1)}

\newcounter{MyPropertyCounter}

\title{
\vspace{-10mm}Bayesian Exploration:\\
Incentivizing Exploration in Bayesian Games%
\footnote{An extended abstract of this paper was published in
\emph{ACM Conf. on Economics and Computation (ACM-EC)}, 2016.
\newline \indent
The working paper has been available at {\tt https://arxiv.org/abs/1602.07570} since Feb 2016.
Revisions focused on presentation; all results (except Appendix~\ref{sec:LB}) have been present since the initial version.
\newline \indent
This research was done while Y. Mansour was a Principal Researcher at Microsoft Research (Herzliya, Israel), and Z.S.Wu was a student at University of Pennsylvania  and a research intern at Microsoft Research NYC.
\newline \indent
The authors wish to thank Dirk Bergemann, Yeon-Koo Che, Shaddin Dughmi, Johannes Horner, Bobby Kleinberg, and Stephen Morris for stimulating discussions on incentivizing exploration and related topics. Also, we are grateful to reviewers of \emph{ACM-EC 2016} and \emph{Operations Research} for numerous helpful suggestions.
 }}

\newcommand{\email}[1]{Email: {\tt #1}.}

\author{
\hspace{-0.75cm} \rule{0.0in}{0pt}
Yishay Mansour
\thanks{Tel Aviv University, Tel Aviv, Israel. \email{mansour@tau.ac.il}}
\and
\hspace{-0.75cm} \rule{0.0in}{0pt}
Aleksandrs Slivkins\thanks{Microsoft Research, New York, NY, USA. \email{slivkins@microsoft.com}}
\and
\hspace{-0.75cm} \rule{0.0in}{0pt}
Vasilis Syrgkanis\thanks{Microsoft Research, Cambridge, MA, USA. \email{vasy@microsoft.com}}
\and
\hspace{-0.75cm} \rule{0.0in}{0pt}
Zhiwei Steven Wu\thanks{Carnegie Mellon University, Pittsburgh, PA, USA. \email{zstevenwu@cmu.edu}}
\hspace{-0.75cm} \rule{0.0in}{0pt}
}

\begin{document}

\date{First version: February 2016\\This version: April 2021}
\maketitle

\vspace{-5mm}
\begin{abstract}
We consider a ubiquitous scenario in the Internet economy when individual decision-makers (henceforth, \emph{agents}) both produce and consume information as they make strategic choices in an uncertain environment. This creates a three-way tradeoff between \emph{exploration} (trying out insufficiently explored alternatives to help others in the future), \emph{exploitation} (making optimal decisions given the information discovered by other agents), and  \emph{incentives} of the agents (who are myopically interested in exploitation, while preferring the others to explore). We posit a principal who controls the flow of information from agents that came before to the ones that arrive later, and strives to coordinate the agents towards a socially optimal balance between exploration and exploitation, not using any monetary transfers. The goal is to design a recommendation policy for the principal which respects agents' incentives and minimizes a suitable notion of \emph{regret}. We extend prior work in this direction to allow the agents to interact with one another in a shared environment: at each time step, multiple agents arrive to play a~\emph{Bayesian game}, receive recommendations, choose their actions, receive their payoffs, and then leave the game forever. The agents now face two sources of uncertainty: the actions of the other agents and the parameters of the uncertain game environment.

Our main contribution is to show that the principal can achieve constant regret when the utilities are deterministic (where the constant depends on the prior distribution, but not on the time horizon), and logarithmic regret when the utilities are stochastic. As a key technical tool, we introduce the concept of {\em explorable actions}, the actions which some incentive-compatible policy can recommend with non-zero probability. We show how the principal can identify (and explore) all explorable actions, and use the revealed information to  perform optimally. In particular, our results significantly improve over the prior work on the special case of a single agent per round, which relies on assumptions to guarantee that all actions are explorable. Interestingly, we do not require the principal's utility to be aligned with the cumulative utility of the agents; instead, the principal can optimize an arbitrary notion of per-round reward.

\end{abstract}

\newpage

\tableofcontents
\vfill
\newpage

\section{Introduction}
\label{sec:intro}


A common phenomenon of the Internet economy is that individual decision-makers (henceforth, \emph{agents}) both produce and consume information as they make strategic decisions in an uncertain environment. Agents produce information through their selection of actions and the resulting outcomes.%
\footnote{This information can be collected explicitly, \eg as reviews on products, or
  implicitly, \eg by observing the routes chosen by a driver and the associated driving times via a GPS-enabled device.}
Agents consume information from other agents who made similar choices in
the past, when and if such information is available, in order to
optimize their utilities.  The collection and dissemination of information relevant to agents' decisions can be instrumented on a very large scale. Numerous online services do it to provide recommendations concerning various products, services and experiences: movies to watch, products to buy, restaurants to dine in, and so forth.

The main issue for this work is that the agents tend to be myopic, optimizing their own immediate reward, whereas the society would benefit if they also explore new or insufficiently explored alternatives. While the tension between acquisition and usage of information is extremely well-studied (under the name of~\emph{exploration-exploitation tradeoff}), a crucial new dimension here is the~\emph{incentives} of the agents: since the agents are self-interested, one cannot expect agents to explore only because they are asked to do so. This creates a new intriguing \emph{three-way tradeoff} between exploration, exploitation and incentives. To study this problem, we introduce a {\em principal}, who abstracts the society or a recommendation system, and whose goal is to optimize some social objective function such as the social welfare.  The principal can control the flow of information from past experiences to the agents, but is not allowed to use  monetary transfers. Absent any new information, an agent will only perform the a priori better action, resulting in no exploration. Full transparency --- revealing to an agent all information currently known by the principal --- is not a good solution, either, because then an agent would only exploit. The goal is to understand how the principal can induce sufficient exploration for achieving near-optimal outcomes.

The prior work on this problem~\citep{Kremer-JPE14,ICexploration-ec15-conf} has a crucial limitation: given the principal's recommendation, the utility of a given agent is assumed to be unaffected by the choices of other agents. This is often not the case in practice. We lift this restriction, and allow the agents to affect one another. Informally, we posit that the agents operate in a shared environment, and an agent's decision can affect this environment for a limited period of time.






\xhdr{Motivating example.}
Let us consider a motivating example based on traffic routing. Consider a GPS-based navigation application such as Waze that gives each user a recommended driving route based on the current traffic conditions in the road network (consuming information received from other drivers), and uses his GPS signal to monitor his progress and the traffic conditions along the route (producing information for future recommendations). A natural goal for the navigation system (the ``principal" in this example) is to minimize the average delay experienced by the users.

If the application is used by only a few drivers, we can view each user in isolation. However, once the application serves a substantial fraction of drivers in the region, there is a new aspect: recommendations may impact the traffic conditions experienced by other drivers. For example, if the application suggests to all the users in a given region to move to a certain lightly loaded route, then when/if they all follow the recommendation, the route may become highly congested. The uncertainty on how congestion affects delays on various routes can be reduced via \emph{exploration}.

A natural simplified model for this scenario is that at each round, multiple drivers arrive and interact through a \emph{routing game} on fixed road network. In this game, the agents simultaneously choose their routes, and the delay at each link is determined by the \emph{load}: the number of routes that use this link. Typically one assumes a parameterized Bayesian model, where the delay at a particular link is a known function of the load and a vector of parameters. The parameter vector is initially not known, but comes from a known Bayesian prior. Thus, we have a Bayesian game among the drivers. Furthermore, there is a principal that recommends routes to drivers: before each round, it recommends a route to each driver that arrives in this round, and observes the associated delays. Since the drivers are not obliged to follow the recommended routes, the recommendations must be ``compatible'' with drivers' incentives, in the sense that they must form a correlated equilibrium in the routing game. At each time-step, the principal has a two-pronged goal: to minimize the driving times (\emph{exploit)} and also to obtain information about the unknown parameters (\emph{explore}) for the sake of giving better recommendations in the future. The challenge is to find an optimal balance between exploration and exploitation under the constraint that even the exploration must be compatible with drivers' incentives.




\asmargincomment{moved all stuff on "results" into the next xhdr}
\xhdr{Our model.}
We put forward a model, called \emph{Bayesian Exploration}, which captures the essence of this motivating example (as well as several other examples discussed below). Initially, some realized state of nature $\trueState$ is selected from a known prior distribution, and never changes since then. There are $T$ rounds. In each round, a new set of $n$ agents arrive and play a game (the same game in all rounds): each agent selects an action, and receives utility determined by the joint action and the state $\trueState$. There is a single principal which, in each round, recommends an action to each participating agent, and observes the chosen actions and the resulting utilities. The recommendations must be Bayesian incentive-compatible (\emph{BIC}), so that the agents are interested in following them.
The principal's utility, a.k.a. \emph{reward}, is also determined by the joint action and the state $\trueState$. The goal of the principal is to maximize her cumulative reward over all rounds, in expectation over the prior (henceforth, \emph{expected reward}). We consider two versions, depending on whether the utilities are ``deterministic" or ``stochastic", \ie whether the utilities of a given joint action are fixed over time (given $\trueState$), or drawn from a fixed distribution.




The major new contribution of our work over \citep{Kremer-JPE14,ICexploration-ec15-conf} is the introduction of multiple agents at each round, and considering an interaction through an arbitrary Bayesian game. This additional interaction allows us to abstract a new host of scenarios, and introduces a new level of complexity to the framework. In particular, the principal now has a variety of joint actions it can recommend, and in many cases it has to randomize its recommendation in order maintain incentive-compatibility. (Technically, in each round the principal needs to select a Bayes-correlated equilibrium \citep{BM-econometrica13}, whereas with a single agent in each round it suffices to select a single action, which is a much simpler object.)

Another important feature of our model is that the utilities of the principal may be unaligned with the (cumulative) utility of all agents. For example, the principal may be interested in some form of fairness, such as maximizing the minimum utility in each round. Such objective is very different from the sum of all utilities. At the extreme, the principal might want to {\em minimize} the cumulative utility. Interestingly, while the BIC constraint generally limits the principal's ability to harm the agents, the principal may still be able to significantly lower the Bayesian-expected social welfare, compared to the worst Bayes-Nash equilibrium that would exist without the principal. (In fact, this effect can be achieved even in a single round, see \citet{PriceOfMediation-DM14}.) This is in stark contrast to the case of single agent per round, where the Bayesian-expected reward of any BIC policy must be at least that of the a-priori best action.

The model is stylized in that we consider the same game in each round and make standard economic assumptions; we discuss this more in Conclusions.

\xhdr{Results and techniques.}
We design policies for the principal that are near-optimal compared to the \emph{best-in-hindsight policy}: a BIC policy which maximizes the expected reward for a particular problem instance. In fact, we compete against a stronger benchmark: optimal expected reward achieved in any one round by any BIC policy.

Our first-order result is a BIC policy whose time-averaged expected reward asymptotically approaches the benchmark. Further, we strive to minimize a version of \emph{Bayesian regret}: the difference between the benchmark and the expected reward of the principal. Following the literature on regret minimization, we are mainly interested in the asymptotic dependence of regret on the time horizon $T$. We achieve optimal asymptotic dependence on $T$:  constant for deterministic utilities, and logarithmic for stochastic utilities.%
\footnote{Our result for stochastic utilities relies on a slightly altered definition of the benchmark, see Section~\ref{sec:model-results} for details.}
The asymptotic constants depend on the Bayesian prior, which appears inevitable because of the BIC constraints. Our policies are computationally efficient: their running time per round is polynomial in the input size.

The stated dependence on $T$ is very non-trivial because the principal faces a game that changes from round to round: namely, the agents' beliefs about the principal depend on the round in which they arrive, because they know that the principal is learning over time. For the sake of argument, consider a modification in which the agents erroneously believe that the principal is not learning. Then the principal plays the same game over and over again, and constant (resp., logarithmic) dependence on $T$ is much easier to obtain.

Our policy for stochastic utilities does not need to know the distribution of ``noise" added to the utilities. This is somewhat surprising because, for example, one would need to know the noise distribution in order to do a Bayesian update on the noisy inputs. Such \emph{detail-free} properties are deemed desirable in the economics literature.

One important technical contribution concerns the concept of the {\em explorable} joint actions: those that can be explored by some BIC policy in some round with positive probability. In general, not all joint actions are explorable, and the set of explorable joint actions may depend on the state $\trueState$. Our high-level approach is to explore all explorable actions. Thus, we face two challenges. First, we need to identify which joint actions are explorable, and explore them. This has to be done inductively, whereby the utilities revealed by exploring one joint action may enable the policy to explore some others, and so forth until no further progress can be achieved. The second challenge is to compute an optimal (randomized) recommendation after all explorable joint actions have been explored, and prove that it outperforms any other BIC policy.

To address these two challenges, we develop the theory of the single-round game in our setting. Termed the \emph{recommendation game}, it is a generalization of the well-known
\emph{Bayesian Persuasion} game \citep{Kamenica-aer11} where the signal observed by the principal is distinct from, but correlated with, the unknown state of nature. (Here the principal's signal captures the history observed in the previous rounds.) This generalization allows us to argue about comparative statics: what happens if the principal's signal is modified to contain more information relevant to the state $\trueState$. We show that a more informative signal can only improve the optimal reward achievable by the principal, and can only increase the set of ``explorable" joint actions. 

The ``monotonicity-in-information" statements mentioned above (that exploring more joint actions or receiving a more informative signal can only help) are non-trivial and not easy to prove, even if they may seem intuitive. It is worth noting that similar monotonicity-in-information statements in game theory are sometimes intuitive but \emph{false} \citep{Wiggans-86,Kessler-98,Syrgkanis-ec15}. More generally, while our algorithms and analyses involve many notions, steps and statements that seem intuitive, formalizing these intuitions required a multi-layered framework to reason about Bayesian Exploration in precise terms. Building up this framework has been a major part of the overall effort.

With stochastic utilities, we have an additional obstacle: the deviations between the expected utilities and our estimates thereof, albeit small and low-probability, may distort agents' incentives. Moreover, the number of different possible observations from the same joint action becomes exponential in the number of rounds, which blows up the running time if we use the techniques from the deterministic case without modifications. To side-step these obstacles, our analysis carefully goes back and forth between the original problem instance and the corresponding problem instance with deterministic utilities.

Our ``algorithmic characterization" of all explorable joint actions is an important contribution even for the case of single agent per round. While \citet{Kremer-JPE14,ICexploration-ec15-conf} provide necessary and sufficient conditions when there are only two actions, for multiple actions \citet{ICexploration-ec15-conf} only give a sufficient condition which is far from necessary. Our work rectifies the situation and provides an explicit algorithm -- in fact, an explicit BIC policy -- to identify all explorable joint actions.

\xhdr{Additional motivating examples.}
The essential features of Bayesian Exploration --- the tradeoff between exploration, exploitation and incentives, and presence of multiple agents in a shared environment --- arise in a variety of scenarios, in addition to the routing scenario described above.

One scenario concerns coordinating participants in a market. Consider sellers that are selling tickets for a particular sports event on an online platform such as StubHub. Implicitly, these sellers are involved in a game where they set the prices and the buyers can select which tickets to buy. A principal (the same platform or a third party) can suggest to the sellers how to set the prices so as to optimize the sellers' combined revenue or other notion of social optimality. The principal learns about the demand of the buyers (through exploration) in order to help the sellers to set the right prices. The price recommendations need to be incentive-compatible because the sellers need to be convinced to follow them. \OMIT{Similarly, one can imagine a principal that coordinates the \emph{buyers} in a similar setting, recommending bidding strategies that optimize some global notion of ``happiness".}

Another scenario concerns optimizing a shared access to a computing resource. For a simple example, consider a university data center shared by the faculty members, where each user can specify the machine to run his jobs on. The principal can collect the information from  multiple users, so as to learn their typical resource requirements, and recommend which machines they should connect to. As a non-critical component, such recommendation system would be easier to maintain compared to a scheduler which enforces the same resource allocation, \eg it can be ignored when/if it malfunctions.

A third scenario addresses a congestion game in which the agents are incentivized towards homophily. Consider a population of individuals choosing which experience to attend in the near future (\eg go to a movie, watch a video on Youtube, read a news article, or attend a sports event). Often people are interested not only in the inherent quality of the experience, but also in sharing it with a particular group. This group may be quite large, \eg it could include many people that share the same demographic, political views, general tastes, etc. Thus, one could imagine a recommendation service that would coordinate people towards sharing the experiences they would like with the people that they would want to share them with. Such recommendation service would need to \emph{explore} to learn people's tastes, and be compatible with people's incentives.

\OMIT{ 
\xhdr{Organization of the paper.}
Model and results are described in Section~\ref{sec:model}. In Section~\ref{sec:warmup}, we solve a simplified version of the problem, and along the way develop some essential tools. Section~\ref{sec:det} presents the main result for the deterministic utilities. Stochastic utilities are treated in Section~\ref{sec:noisy}. A useful fact about the single-round game (a major tool in the preceding sections) is summarized and proved in Section~\ref{app:coupled-signals}. Conclusions are in Section~\ref{sec:conclusions}. In particular, we discuss some limitation of our model, and suggest directions for follow-up work.
} 

\section{Related Work}
\label{sec:related-work}

The study of ``incentivizing exploration" was initiated by~\citet{Kremer-JPE14} and \cite{Che-13}. Most relevantly to this paper, \citet{Kremer-JPE14} ,  introduce a version of Bayesian Exploration with only two actions and without strategic interactions between agents, \ie with a single agent in each round. They provide an optimal policy for deterministic utilities, and a preliminary result for stochastic utilities. \citet{ICexploration-ec15-conf,ICexploration-ec15} obtain optimal regret rate for stochastic utilities, as well as a reduction from an arbitrary non-BIC policy to a BIC one. \citet{Bimpikis-exploration-ms17} consider time-discounted utilities. They achieve ``first-best" utility if the expected utility for one of the actions is known. For the general case, they design an optimal but computationally inefficient mechanism, and propose a computationally efficient heuristic. \citet{Bahar-ec16} enrich the model to allow agents to observe recommendations of their ``friends" in a known social network (but restrict to deterministic utilities and a small number of high-degree nodes). \citet{Frazier-ec14} and \citet{Che-13} study related, but technically different problems, where, resp., the principal can pay the agents and one has continuous information flow and a continuum of agents. While the core challenges in incentivizing exploration arise even with two actions, \citet{ICexploration-ec15-conf,ICexploration-ec15} and \citet{Frazier-ec14} extend their results to multiple actions.

\asmargincomment{rewrote next two paragraphs}
Absent strategic issues, Bayesian Exploration reduces to \emph{multi-armed bandits}, a fundamental and well-studied model for exploration-exploitation tradeoff. A huge literature on bandits has been summarized in several books
\citep{Bubeck-survey12,Gittins-book11,slivkins-MABbook,LS19bandit-book}. More specifically, if the algorithm is not constrained to be BIC and the agents always follow the recommendations, then Bayesian Exloration reduces to multi-armed bandits with stochastic rewards, where action set is $\cA$ and rewards are equal to the principal's utility. This is a well-understood problem, with optimal regret-minimizing algorithms and matching lower bounds
\citep{Lai-Robbins-85,bandits-ucb1,bandits-exp3}.

Bayesian Exploration is closely related to \emph{information design} and \emph{social learning},
two prominent subareas of theoretical economics.
Information design studies the design of information disclosure policies and incentives therein
\citep[see surveys][]{BergemannMorris-survey19,Kamenica-survey19,Dughmi-stoc16}.
In particular, a single round of Bayesian Exploration is a version of the Bayesian Persuasion game \citep{Kamenica-aer11} with multiple agents, where the signal observed by the principal is correlated with the unknown ``state".
Our notion of BIC is a version of Bayes Correlated Equilibrium \citep{BM-econometrica13}.
Social learning studies self-interested agents that interact and learn over time in a shared environment. A version of our setting with long-lived agents and no principal to coordinate them has been studied since \citet{Bolton-econometrica99}, see \citet{Horner-survey16} for a survey. A vast literature on \emph{learning in games} \citep{FLbook} posits reasonable learning dynamics for the agents, as a proxy for their strategic behavior, and studies convergence to a given solution concept. A long line of work on \emph{sequential social learning} posits that agents receive private signals and observe previous agents' actions before choosing their own, see \citet{Golub-survey16} for a survey.

Exploration-exploitation problems with self-interested agents arise in several other, technically incomparable scenarios:
dynamic pricing
    \citep[\eg][]{KleinbergL03,BZ09,DynPricing-ec12,Wang-OR14},
dynamic auctions
    \citep[\eg][]{AtheySegal-econometrica13,DynPivot-econometrica10,Kakade-pivot-or13},
pay-per-click ad auctions
    \citep[\eg][]{MechMAB-ec09,DevanurK09,Transform-ec10-jacm},
human computation
    \citep[\eg][]{RepeatedPA-ec14,Ghosh-itcs13,Krause-www13},
\asedit{and competition between firms
    \citep[\eg][]{bergemann2000experimentation,keller2003price,CompetingBandits-merged}.
A unified perspective on exploration with incentives can be found in Chapter 11.6 in \citet{slivkins-MABbook}.}

\xhdr{Subsequent work.}
\asedit{
The study of Bayesian Exploration has progressed subsequent to the initial version of this paper \citep{ICexplorationGames-ec16-conf}. All work has focused on the version of a single agent per round.
First, the economic model is generalized: \citet{Jieming-multitypes18} study an extension to heterogenous agents with public or private types, and \citet{Bahar-ec19} consider unavoidable information leakage (in a different model compared to \citet{Bahar-ec16}).
Second, results are refined in the basic model with stochastic rewards:
\citet{Selke-BE-2020} characterize the ``price of incentives": loss in performance compared to optimal multi-armed bandit algorithms,
and \cite{Jieming-unbiased18} mitigate economic assumptions and make a direct connection to social learning on networks.
Third, the exploration problem is generalized: \citet{IncentivizedRL} incentivize exploration in reinforcement learning.
}

\OMIT{A related task in reinforcement learning is that of learning the structure of a Markov Decision Process (MDP), which can be done efficiently
\cite{KearnsS02,BrafmanT02}. In our setting, the principle has to uncover the structure of the underlying game, which vaguely resembles the
learning process of uncovering the structure of an MDP. Beyond the fact that the model significantly differ, there is the additional strategic
issue that our principal has to consider. \ymcomment{not sure if we need this}
}

\section{Bayesian Exploration: our model and results}
\label{sec:model}


Our model, called \emph{Bayesian Exploration}, is a game between a principal and multiple agents. The game consists of $T$ rounds, where $T$ is the time horizon.

There is a global parameter $\trueState$, called the \emph{(realized) state}. It is drawn from a Bayesian prior distribution $\prior$ over a finite state space $\Theta$. It is chosen before the first round, and stays the same throughout the game. The state space $\Theta$ and the prior $\prior$ are common knowledge, but the state itself is not revealed neither to the principal, nor to the agents. (But the principal can learn it over time.) Elements of $\Theta$ are called \emph{(feasible) states}, to distinguish them from the realized state $\trueState$.

In each round, there is a fresh set of agents, denoted $[n]:=\{1 \LDOTS n\}$, playing a simultaneous game. Each agent $i$ chooses an action $a_i$ from some finite set $\cA_i$ of possible actions. A tuple $a = (a_i)_{i\in [n]}$ is called a \emph{joint action}.
The set of all possible joint actions, termed the \emph{action set}, is $\cA= \cA_1\times \ldots \times \cA_n$. The utility of each agent $i$ is determined by the joint action $a$ chosen by the agents and the realized state $\trueState$. More formally, it is given by a function $u_i\colon \cA\times\Theta\to [0,1]$, called the \emph{utility function} of agent $i$. The utility of the principal, a.k.a. \emph{reward}, is also determined by the pair $(a,\trueState)$, and given by the \emph{reward function} $f\colon\cA\times \Theta\to [0,1]$.  The action set $\cA$ and the functions $(f; u_1 \LDOTS u_n)$ are the same for all rounds, and are common knowledge.

In each round $t$, the principal recommends a joint action to the agents. Specifically, the round proceeds as follows: a fresh set of $n$ agents arrives; the principal recommends a joint action $a\in\cA$; each agent $i$ only observes his recommended action $a_i$; the agents choose their actions in the simultaneous game, and the utilities are realized. The principal observes the chosen actions and the realized utilities. Agents know $t$, but do not observe what happened in the previous rounds.

The principal commits to an algorithm $\pi$ that proceeds in rounds, so that in each round it outputs a joint action, and then inputs the chosen actions and the realized utilities. This algorithm, called the \emph{iterative recommendation policy}, is known to the agents.

We will now define the~\emph{incentive-compatibility} constraint. Let
$\pi^t$ be the joint action recommended by $\pi$ in round $t$; we interpret it as a
random variable taking values in $\cA$. Let $\cE_{t-1}$ be the
event that the agents have followed principal's recommendations up to
(but not including) round $t$. Here and henceforth, we will use a
standard game-theoretic notation: we will represent a joint action
$a\in\cA$ as a pair $(a_i,a_{-i})$, where $i$ is an agent, $a_i\in\cA$
is its action, and $a_{-i} = (a_1 , \ldots , a_{i-1}, a_{i+1} , \ldots
, a_n)$ is the joint action of all other agents; we will sometimes
write $(a_i,a_{-i};\theta)$ to denote the pair $(a,\theta)$, where $\theta\in\Theta$ is a state. We will
write $\cA_{-i}$ for the set of all possible $a_{-i}$'s. In
particular, $\pi^t_{-i}\in \cA_{-i}$ denotes the joint action of all
other agents chosen by policy $\pi$ in round $t$.


\begin{definition}\label{def:BIC}
An iterative recommendation policy $\pi$ is \emph{Bayesian Incentive
  Compatible (BIC)} if for all rounds $t$ and agents $i\in [n]$ we have
\begin{align}\label{eq:def:BIC}
\E\sbr{ u_i(a_i,\,\pi^t_{-i};\, \trueState) - u_i(a_i',\,\pi^t_{-i};\,\trueState)
    \mid \pi^t_i=a_i,\,\cE_{t-1} }
    \geq 0,
\end{align}
where $a_i,a_i'\in \cA_i$ are any two distinct actions such that
    $\Pr[\pi^t_i=a_i \mid \cE_{t-1}]>0$.
(The probabilities are over the realized state $\trueState\sim \prior$ and the internal randomization in $\pi$.)
\end{definition}

\noindent In words, suppose agent $i$ is recommended to play action $a_i$ in a given round $t$. Assume that all agents followed the recommendations in the previous rounds, and that all agents but $i$ follow the recommendations in round $t$. Then agent $i$ cannot improve his conditional expected utility (given the information available to this agent) by choosing a different action $a_i'$.

Throughout, we assume that the agents follow recommendations of a BIC iterative recommendation policy, so that the agents' behavior is completely specified.

The principal's goal is to optimize its expected total reward, as given by the reward function $f$, subject to the BIC constraint. While we allow $f$ to be arbitrary, one natural version is social welfare, \ie the sum of utilities of all agents in a given round.

\OMIT{ 
\xhdr{Finite support assumption.}
For simplicity of exposition, we assume that the state space $\Theta$ is finite, and the utility distributions $\cD_{(a,\theta)}$, $(a,\theta)\in \cA\times \Theta$  have finite support.
} 

The \emph{utility structure} $\cU$ is a tuple that consists of the action set $\cA$, the state space $\Theta$, the prior $\psi$, the utility functions $(u_1 \LDOTS u_n)$, and the reward function $f$. Note that $\cU$ completely determines the simultaneous game. A problem instance consists of $\cU$ and time time horizon $T$.


\xhdr{Stochastic utilities.} We allow a more general version, termed Bayesian Exploration with \emph{stochastic utilities}. Given the joint action $a\in\cA$ and state $\theta = \trueState$, the vector of all realized utilities is drawn independently from some distribution $\cD_{(a,\theta)}$ over such vectors, so that the expected utility of each agent $i$ and the principal are, respectively, $u_i(a,\theta)$ and $f(a,\theta)$.   All realized utilities lie in $[0,1]$.%
\footnote{Our results extend to a more general case of sub-Gaussian noise; we omit the easy details.} The special case when the realized utilities are always equal to their expectations is termed Bayesian Exploration with \emph{deterministic utilities}.

Each distribution
    $\cD_{(a,\theta)}$
is determined by the utility tuple
    \[ \cU_{(a,\theta)} = \rbr{ u_1(a,\theta) \LDOTS u_n(a,\theta);\;f(a,\theta) }. \]
(We use this assumption to prove Lemma~\ref{lm:model-benchmarks}.)
We assume that the mapping from utility tuples $\cU_{(a,\theta)}$ to distributions $\cD_{(a,\theta)}$ is known to the agents, so that they can form Bayesian posteriors. Whether this mapping is known to the mechanism is irrelevant to our results.

A problem instance of Bayesian Exploration with stochastic utilities consists of the utility structure $\cU$, time horizon $T$, and the distributions
    $\cD_{(a,\theta)}$, $(a,\theta)\in \cA\times \Theta$.

\xhdr{Computational model.} Our results include bounds on the per-round running time of our algorithms. These bounds assume \emph{infinite-precision arithmetic} (infinite-precision arithmetic operations can be done in unit time), and \emph{continuous random seed} (a number can be drawn independently and uniformly at random from the unit interval in unit time). Such assumptions are commonly used in theoretical computer science to simplify exposition.

From the computational point of view, an iterative recommendation policy inputs the utility structure. We assume the utilities are represented generically, as a $(n+1)\times |\cA|\times |\Theta|$ table. We obtain per-round running times that are \emph{polynomial in the input size}.

\xhdr{Discussion.}
One can consider a more general version of Bayesian Exploration in which the principal can send arbitrary messages to agents. However, the restriction of messages to recommended actions is without loss of generality (see Remark~\ref{rem:BIC}).

\asedit{We fix the time horizon $T$ only for ease of presentation. Knowing $T$ is irrelevant to the definition of BIC and hence to the agents, and not needed in our solution for deterministic utilities. For stochastic utilities, knowing $T$ can be mitigated via standard tricks from multi-armed bandits}.%
\footnote{{First, one could use an upper bound $T'\geq T$ instead of $T$. It is immediate from our analysis that such policy is BIC and suffers only $\log(T')$ in the regret bound in Theorem~\ref{thm:model-noisy}. So, even a crude upper bound $T'$ would suffice; this would probably be a preferred solution in practice. Second, one could starts with a BIC iterative recommendation policy parameterized by a known time horizon, and transform it into a policy that runs indefinitely and attains similar performance. Specifically, we partition time into consecutive phases $i=1,2,\,\ldots$, each of predetermined duration $T_i$ rounds, and restart the policy in the beginning of each phase with time horizon $T_i$. It is easy to see that such policy is BIC.  Durations $T_i = 2^{2^i}$ suffice for our purposes, to guarantee $O(\log T)$ dependence on regret in each round $T$.}}

Formally, we will assume that every BIC iterative recommendation policy $\pi$ is well-defined and BIC for all rounds $t\in \NN$, even if it is tailored to a particular time horizon $T$. This is without loss of generality: \eg one can extend $\pi$ by setting $\pi^t=\pi^T$ for all rounds $t>T$.

\subsection{Statement of the main results}
\label{sec:model-results}

Given an iterative recommendation policy $\pi$, the total expected reward is denoted
    \[ \REW(\pi) = \textstyle \sum_{t=1}^T \REW_t(\pi^t),\quad \REW_t(\pi) = \E\sbr{ f(\pi^t,\trueState)}.\]
Here $\REW_t(\pi)$ is the expected reward of $\pi$ in a given round $t$. The expectation is taken over the realized state $\trueState$, the internal randomness in $\pi$, and possibly the noise in the rewards.

\xhdr{Deterministic utilities.}
We compete with all BIC iterative recommendation policies with infinite time horizon; the class of all such policies is denoted $\PiBIC$. Our benchmark  is the best expected reward achieved in any one round by any policy in a given policy class $\Pi$:
\begin{align}\label{eq:model-benchmark}
\opt(\Pi) := \sup_{t\in \NN,\, \pi\in\Pi} \REW_t(\pi).
\end{align}
Our main result is a BIC policy whose time-averaged expected reward is close to $\opt(\PiBIC)$.

\begin{theorem}\label{thm:model-det}
Consider Bayesian Exploration with deterministic utilities. There exists a BIC iterative recommendation policy $\pi$ whose Bayesian regret with respect to $\PiBIC$ satisfies
\begin{align}\label{eq:thm:model-det}
T\cdot\opt(\PiBIC) - \REW(\pi) \leq C,
\end{align}
where $C$ depends only on the utility structure, but not on the time horizon.%
\footnote{\asedit{Specifically, $C = |\cA|^2\cdot |\Theta|/\psig^*$, where $\psig^*$ is spelled out in Remark~\ref{rem:prior-dep}.}}
The per-round running time of $\pi$ is polynomial in the generic input size.
\end{theorem}


\xhdr{Stochastic utilities.} Our results for stochastic utilities compete against a slightly restricted class of BIC policies. For a given parameter $\delta> 0$, a policy is called \emph{$\delta$-BIC} if it satisfies a stronger version of Definition~\ref{def:BIC} in which right-hand side of \eqref{eq:def:BIC} is $\delta$ rather than $0$.%
\footnote{For ease of notation, we extend this definition to $\delta=0$, in which case ``$\delta$-BIC" is the same as ``BIC".}
 The class of all such policies is denoted $\PiBIC[\delta]$. We construct a BIC policy whose time-averaged expected reward is close to $\opt(\PiBIC[\delta])$.

\begin{theorem}\label{thm:model-noisy}
Consider Bayesian Exploration with stochastic utilities. For any given $\delta>0$, there exists a BIC iterative recommendation policy $\pi$ whose Bayesian regret with respect to $\PiBIC[\delta]$ satisfies
\begin{align}\label{eq:thm:model-stoch}
T\cdot \opt(\PiBIC[\delta]) - \REW(\pi) \leq C_\delta\cdot \log T,
\end{align}
where $C_\delta$ depends only on the utility structure and the parameter $\delta$, but not on the time horizon.%
\footnote{\asedit{The exact guarantee is spelled out in \refeq{eq:noisy-final}}.}
The per-round running time of $\pi$ is polynomial in
the generic input size.
Policy $\pi$ does not input the parameterized utility distributions $\cD_{(a,\theta)}$, $(a,\theta)\in\cA\times\Theta$.
\end{theorem}


The ``structure-dependent" constants in Theorem~\ref{thm:model-det} and  Theorem~\ref{thm:model-noisy} are in line with prior work on Bayesian Exploration: \eg all results in \citet{ICexploration-ec15-conf,ICexploration-ec15} (which are for the special case of single agent per round) have similar `structure-dependent" constants. As in prior work, it is unclear whether the dependence on particular parameters is necessary or optimal.%
\footnote{\asedit{A recent follow-up paper \citep{Selke-BE-2020} partially characterizes optimal dependence on parameters in the special case of one agent per round, assuming that the prior is independent across actions. They focus on polynomial vs. exponential dependence on the number of actions and the strength of beliefs.}}

In Appendix~\ref{sec:LB}, we prove that the $O(\log T)$ regret rate in Theorem~\ref{thm:model-noisy} is essentially optimal for Bayesian Exploration with stochastic utilities, even with a single agent per round and only two actions. This is in line with the fact that $O(\log T)$ regret rate is optimal in multi-armed bandits \citep{Lai-Robbins-85,bandits-ucb1}.  Our proof builds on the negative result for bandits, and also invokes a positive result for Bayesian Exploration from \citet{ICexploration-ec15-conf,ICexploration-ec15}.



\subsection{Discussion of the benchmarks}
\label{sec:model-benchmarks}

\OMIT{ 
We prove a stronger version of Theorem~\ref{thm:model-noisy}, in which $\opt(\PiBIC[\delta])$ is replaced with a similar benchmark for \emph{deterministic} utilities. Given a problem instance with stochastic utilities, we consider the \emph{deterministic instance}: a version of the original problem instance where all utilities are deterministically equal to the corresponding expected utilities in the original instance. For a policy class $\Pi$, we define $\optdet(\Pi)$ as the value of $\opt(\Pi)$ for the deterministic instance. We focus on the class of all iterative recommendation policies that are $\delta$-BIC for the deterministic instance, denoted $\PiBICdet[\delta]$. The new benchmark is defined as
    $\optdet(\PiBICdet[\delta])$.
It is at least as strong as $\opt(\PiBIC[\delta])$, in the sense that
$\optdet(\PiBICdet[\delta])\geq \opt(\PiBIC[\delta])$ for all $\delta>0$.
} 

Recall that our generic benchmark $\opt(\Pi)$ considers the optimal reward in any one round, and does not have a built-in time horizon. Therefore, it is at least as strong as the optimal \emph{time-averaged} reward for the same policy class,
    $\optave(\Pi) := \sup_{\pi\in\Pi,\; T\in\NN} \REW(\pi)/T$.
While the latter benchmark can be achieved by some policy in $\Pi$ (up to an arbitrarily small precision), this is not trivially the case for $\opt(\Pi)$. However, our theorems imply that these two benchmarks are equal: for deterministic rewards with $\Pi = \PiBIC$, and for stochastic rewards with $\Pi=\PiBIC[\delta]$, $\delta>0$.

Our benchmarks $\opt(\PiBIC)$ and $\opt(\PiBIC[\delta])$ are, in general, weaker than the \emph{best-action benchmark} (the expected reward of the best action), a standard benchmark in multi-armed bandits. In our notation, this benchmark equals
    $\E_{\trueState\sim\prior}\sbr{ \max_{a\in\cA} f(a,\trueState) }$.
BIC policies cannot always match the best action because not all actions are explorable, and even explorable actions are not necessarily implementable with high probability by a BIC policy. We know, however, that the benchmark $\opt(\PiBIC)$ for deterministic utilities does reduce to the best action in the special case when there is a single agent per round and the principal's reward coincides with the agent's utility, as long as all actions are explorable (see Corollary~\ref{cor:deterministic-benchmark-characterization}).%
\footnote{For this special case, \citet{ICexploration-ec15-conf,ICexploration-ec15} provide sufficient conditions when all arms are explorable. In particular, this happens when the prior is independent across actions, and for each action the prior on the mean reward has full support on $[0,1]$. \citet{ICexploration-ec15-conf,ICexploration-ec15} also provide another sufficient condition which allows for correlated priors. However, we cannot immediately apply these sufficient conditions to help characterize our benchmark $\PiBICdet[\delta]$ for stochastic rewards. This is because the results in \citet{ICexploration-ec15-conf,ICexploration-ec15} only guarantee that all actions are explorable by some BIC policy, but not necessarily by a $\delta$-BIC policy for a given $\delta>0$.}

Our benchmark for deterministic utilities can be characterized as follows: $\opt(\PiBIC)$ is the maximal per-round expected reward of a BIC policy that has access to the utilities of all explorable joint actions, see Theorem~\ref{thm:deterministic-benchmark-characterization} for a precise statement.

\OMIT{Similarly, $\optdet(\PiBICdet[\delta])$ is the maximal per-round expected reward of a $\delta$-BIC policy that has access to the utilities of all joint actions that are explorable by a $\delta$-BIC policy.}

\OMIT{
We also assume that the utility structure is \emph{well-separated}, in the following sense: there is a parameter $\zeta>0$ such that for each agent $i$, each joint action $a\in\cA$, and any two states $\theta,\theta'\in\Theta$ with
    $u_i(a,\theta) \neq u_i(a,\theta')$
we have that
    $|u_i(a,\theta) - u_i(a,\theta')|\geq\zeta$.
} 


%

\section{Warm-up and tools}
\label{sec:warmup}

As a warm-up, let us consider the version with deterministic utilities, and focus on a relatively simple scenario when a BIC iterative recommendation policy explores all joint actions, and then \emph{exploits} (in the sense that we make precise later). We show that this policy can achieve optimal per-round performance once all joint actions are explored. Recall that our benchmark is $\opt(\PiBIC)$, where $\PiBIC$ is the class of all BIC iterative recommendation policies, and $\opt(\cdot)$ is defined in \eqref{eq:model-benchmark}.

\begin{lemma}\label{lm:warmup}
Consider Bayesian Exploration with deterministic utilities. Let $\pi$ be a BIC  iterative recommendation policy that explores all joint actions by a fixed time $T_\pi\leq T$. Then there exists a BIC policy $\pi'$ which coincides with $\pi$ before round $T_\pi$, and achieves expected reward at least
    $\opt(\PiBIC)$
in all subsequent rounds. Therefore,
$f(\pi') \geq (T-T_\pi)\;\opt(\PiBIC)$.
\end{lemma}

\noindent While very intuitive, this result is surprisingly technical to prove from scratch. Essentially, one needs to specify what ``exploitation" means in this context, and argue that this notion of exploitation is BIC and can only benefit from having full information about the utility structure. Thus, we develop a framework to reason about this, which will be an essential toolbox throughout the paper. More specifically, we define and analyze a game which captures a single round of Bayesian Exploration, and formulate a framework to combine BIC ``subroutines" into a BIC iterative recommendation policy. A proof of Lemma~\ref{lm:warmup} using these tools is in the very end of this section.

While Lemma~\ref{lm:warmup} relies on the ability to explore all joint actions, this ability is not guaranteed. This can be seen even in the special case of a single agent per round and only two actions. For this special case, \citet{Kremer-JPE14} and \citet{ICexploration-ec15} present necessary and sufficient conditions under which all actions are explorable, as well as simple examples when these conditions fail. \citet{ICexploration-ec15} also provides sufficient conditions for the version with a single agent per round and an arbitrary number of actions.

\subsection{The recommendation game: a single round of Bayesian Exploration}
\label{sec:rec-game}

We view a single round of Bayesian Exploration as a stand-alone game between the principal and the agents, termed the \emph{recommendation game}. Here the principal observes an auxiliary ``signal", which represents the information received in the previous rounds (and possibly also the internal random seed). Then the principal recommends a joint action, and the agents choose their actions.

Formally, the recommendation game is a version of the Bayesian Persuasion game \citep{Kamenica-aer11} with multiple agents. Unlike the original Bayesian Persuasion game, in our version the signal observed by the principal is distinct from (but correlated with) the state.

For a problem instance of Bayesian Exploration with a given utility structure, the corresponding \emph{recommendation game} proceeds as follows:

\begin{OneLiners}
\item the state $\trueState$ is drawn from a Bayesian prior distribution $\psi$ over $\Theta$;
\item the principal observes a \emph{signal} $S$, then recommends action $a_i\in\cA_i$ for each agent $i$;
\item the agents choose their actions in the simultaneous game;
\item the principal and the agents receive utilities according to the utility structure.
\end{OneLiners}
\noindent Signal $S$ is an arbitrary random variable with finite support $\cX$ (the elements of $\cX$ are called \emph{feasible signals}). The signal can be correlated with the state: formally, $S$ and $\trueState$ are random variables on the same probability space. The \emph{signal structure} associated with $S$ is the tuple $(\Theta,\cX,\joint)$, where $\joint$ is the joint distribution of $(S,\trueState)$. An important special case is the \emph{empty signal}: one which always takes the same value. Such signal will be denoted as $S=\bot$.


The utility structure and the signal structure are common knowledge. The realized state $\trueState$ is not revealed to the principal (other than through the signal $S$). Each agent $i$ only observes his own recommendation $a_i$; it does not observe the state $\trueState$, the signal $S$, or the other recommendations

The principal commits to a \emph{recommendation policy}: a randomized
mapping $\pi \colon \cX \to \cA$, which takes as input a feasible signal and outputs an action for each agent. The corresponding incentive-compatibility constraint is defined as follows:

\begin{definition}\label{def:rec-BIC}
Given signal $S$, a recommendation policy $\pi$ is~\emph{Bayesian incentive compatible (BIC)} if for each agent $i\in [n]$ and any two distinct actions $a_i,a_i'\in \cA_i$ such that
    $\Pr[\pi_i(S)=a_i]>0$
we have
\begin{align}\label{eq:def:rec-BIC}
\E \sbr{ u_i(a_i,\,\pi_{-i}(S);\, \trueState) - u_i(a_i',\,\pi_{-i}(S)\,;\trueState) \mid \pi_i(S)=a_i }
    \geq 0.
\end{align}
\end{definition}

In words, whenever agent $i$ is recommended to play some action $a_i$, he could not improve his expected utility (given the information available to this agent, and assuming that all other agents follow the recommendations) by choosing a different action $a_i'$.
We assume that the agents follow the recommendations of a BIC policy, so that the expected reward is well-defined.

For the recommendation game, the distinction between stochastic and deterministic utilities is unimportant (for statements that only involve expected utilities and rewards). In particular, a given recommendation policy is BIC for stochastic utilities if and only if it is BIC for the corresponding problem instance with deterministic utilities.

For the computational results, we assume that the joint distribution of $(S,\trueState)$ is given explicitly, as a $|\cX|\times|\Theta|$ table of probabilities.

\begin{remark}
Consider an iterative recommendation policy $\pi$.  W.l.o.g., the internal random seed $\omega$ of policy $\pi$ is chosen once, before the first round, and persists throughout. Let $H_t$ be the \emph{history} up to round $t$: the chosen joint actions and the realized utilities over all past rounds. Then policy $\pi$ can be represented as a sequence
    $(\pi^{(t)}:\, t\in\NN)$,
where for each round $t$, $\pi^{(t)}$ is a randomized mapping from $S_t=(\omega,H_t)$ to joint actions (called the \emph{restriction} of $\pi$ to round $t$). Each round $t$ can be seen as a recommendation game with signal $S_t$, and $\pi^{(t)}$ is a recommendation policy in this game. It is easy to see that $\pi$ is BIC if and only if $\pi^{(t)}$ is BIC for all $t$.
\end{remark}

\begin{remark}
\label{rem:BIC}
The notion of Bayesian Incentive Compatibility in the recommendation game is closely connected to the notion of Bayes Correlated Equilibrium \citep{BM-econometrica13}: modulo the differences in terminology, the former is a special case in which the agents do not receive private signals. In particular, it follows that a BIC recommendation policy always exists.

One can consider a more general version of the recommendation game in which the principal can send arbitrary messages to agents. Then the principal commits to a \emph{messaging policy}: a randomized mapping which inputs the signal $S$ and outputs a message $m_i$ for  each agent $i$. However, such messaging policies can without loss of generality be restricted to recommendation policies. More precisely, suppose a messaging policy $\tau$ induces a Bayes Nash Equilibrium $\rho$ (which, in our notation, is a randomized mapping that inputs the joint message $(m_1 \LDOTS m_n)$ and outputs a joint action $a\in\cA$). Then the composition $\rho\oplus \tau$ is a randomized mapping that inputs signal $S$ and outputs a joint action, \ie a recommendation policy. According to \cite{BM-econometrica13}, this composition is BIC.%
\footnote{This follows from the "only if" direction of Theorem 1 in \cite{BM-econometrica13}. In their notation, we use the special case where the information structure $S$ is empty, and $S^*$ is induced by the messaging policy $\tau$. Then the ``decision rule" $\pi$ in the theorem corresponds to the recommendation policy $\rho\oplus \tau$.}
Thus, the principal can use $\rho\oplus \tau$ instead of $\tau$.
\end{remark}

\subsection{Optimality in the recommendation game}
\label{sec:rec-game-props}

Given a recommendation policy $\pi$ for signal $S$, its expected reward is
\begin{align*}
    \REW(\pi)
        = \E\sbr{ f(\pi(S),\trueState) }
        = \E_{(s,\theta)\sim \joint} \sbr{ f(\pi(s),\theta) }.
\end{align*}
The expectation is over the joint distribution of signal $S$ and realized state $\trueState$, and the internal randomness in policy $\pi$. The \emph{optimal reward} given signal $S$ is defined as
\begin{align*}
    \REW^*[S] = \sup_{\text{BIC recommendation policies $\pi$ for $S$}} \REW(\pi).
\end{align*}

\noindent Thus, it is the largest expected reward achievable in a recommendation game with signal $S$.%
\footnote{We write $\REW^*[S]$ rather than $\REW^*(S)$ to emphasize that the optimal reward  is a function of the random variable $S$, rather than a function of a particular realization of this random variable. We will use a similar notation $\mathtt{function}[S]$ for some other functions of signal $S$ as well, \eg for the polytope $\bic[S]$.}
A BIC policy $\pi$ in this game will be called \emph{optimal} for $S$ if
    $\REW(\pi) = \REW^*[S]$.

We will represent the problem of finding an optimal recommendation policy $\pi$ as a linear program (henceforth, \emph{LP}). We represent $\pi$ as a set of numbers
    $x_{a,s} = \Pr[\pi(s)=a]$,
for each joint action $a\in \cA$ and each feasible signal $s\in \cX$. These numbers, termed the \emph{LP-representation} of $\pi$, will be the decision variables in the LP. The linear program is as follows:

\begin{framed}
\vspace{-5mm}
\begin{align*}
& \mbox{maximize }
&\sum_{a\in \cA,\; s\in\cI,\; \theta\in \Theta}
    \prior(\theta)
    \cdot  \joint(s\mid\theta)
    \cdot x_{a,s}\cdot f(a,\theta) \\
& \mbox{for all } i\in [n], a_i, a_i'\in \cA_i,
&\sum_{a_{-i}\in \cA_{-i},\; s\in\cI,\; \theta\in \Theta}
    \prior(\theta)
    \cdot\joint(s\mid\theta)
    \cdot \left( u_i(a_i, a_{-i}; \theta) - u_i(a_i', a_{-i}; \theta) \right) \cdot x_{a, s} \geq 0\\
&\mbox{for all } s\in \cI, a\in\cA
    & x_{a,s} \geq 0
        \quad\text{and}\quad
        \textstyle \sum_{a\in\cA} x_{a, s} = 1
\end{align*}
\vspace{-7mm}
\end{framed}

\noindent In the above LP, $\prior(\theta)$ stands for $\Pr[\trueState=\theta]$, and $\joint(s\mid\theta)$ stands for
    $\Pr[S=s \mid \trueState = \theta] $.
The objective is $\REW(\pi)$, written in terms of the LP-representation, and the first constraint states that $\pi$ is BIC.%
\footnote{For a particular agent $i$ and actions $a_i,a'_i\in \cA_i$, the BIC constraint in Definition~\ref{def:rec-BIC} states that $\Pr[\cE]>0$ implies $\E[W\mid \cE]\geq 0$, where
    $W = u_i(a_i,\pi_{-i}(S),\trueState) - u_i(a'_i,\pi_{-i}(S),\trueState)$
and
    $\cE = \{\pi_i(S)=a_i\}$.
This is equivalent to $\E[{\bf 1}_\cE\cdot W]\geq 0$. And $\E[{\bf 1}_\cE\cdot W]$ is precisely the left-hand side of the first constraint in the LP.}
The feasible region of this LP is polytope in $\RR^{|\cI| \times |\cA|}$, which we denote by $\bic[S]$.


\begin{claim}\label{cl:LP-rep}
A policy $\pi$ is BIC if and only if its LP-representation lies in $\bic[S]$.
\end{claim}

\begin{corollary}\label{cor:LP-rep}
An optimal recommendation policy exists. Given the utility structure and the signal structure, an optimal recommendation policy and the optimal reward $\REW^*[S]$ can be computed in time polynomial in $|\cA|\cdot |\cX|\cdot |\Theta|$. Further, a convex combination of BIC policies is also BIC.
\end{corollary}

Let us study a recommendation game with a signal that corresponds to exploring a given (possibly randomized) subset $B\subset \cA$ of joint actions. Formally, $B$ will be a \emph{$2^\cA$-valued signal}: a signal whose values are subsets of $\cA$; \eg its realization may depend on the realized state $\trueState$.

Let us consider the signal which corresponds to exploring all joint actions in $B$. This signal consists of all relevant utilities, and also includes $B$ itself:
\begin{align}\label{eq:signalB-def}
\signal(B) :=
\left(\; B;\quad
    \left(f(a,\trueState); u_1(a,\trueState) \LDOTS u_n(a,\trueState)\right):
    \; a\in B
\right).
\end{align}
Note that it is a random variable, because it depends on random variables $B$ and $\trueState$.

Despite a rather complicated definition, $S=\signal(B)$ is just a signal in a recommendation game. In particular, we can consider an optimal recommendation policy given this signal. By Corollary~\ref{cor:LP-rep}, such a policy exists, and its LP-representation (and its expected reward) can be computed in time polynomial in  $|\Theta|\cdot |\cA|\cdot |\support(S)|$.


We prove that if a BIC iterative recommendation policy is restricted to joint actions in $B$, then its expected per-round reward cannot exceed $\REW^*\sbr{\signal(B)}$ (the proof requires additional machinery, developed in the next subsection).

\begin{claim}\label{cl:restricted-subset}
Consider Bayesian Exploration with stochastic utilities. Let $B\subset\cA$ be a $2^\cA$-valued signal (\ie a randomized set of joint actions). Let $\pi$ be a BIC iterative recommendation policy such that before a given round $t$ it can only choose joint actions in $B$. Then
$\REW_t(\pi) \leq \REW^*\sbr{ \signal(B) }$.
\end{claim}

\subsection{Recommendation game with coupled signals}
\label{sec:coupled}

Throughout, we consider a fixed utility structure $\cU$, but allow the signal $S$ to vary from one game instance to another. We suppress $\cU$ from our notation, but make the dependence on $S$ explicit. To study how the properties of a game depend on a particular signal $S$, we consider multiple signals that are random variables in the same probability space as the true state $\trueState$; such signals will be called \emph{coupled}. While each of these signals corresponds to a separate game instance, we will refer to all these game instances jointly as \emph{recommendation game with coupled signals}.

We prove that $\REW^*[S]$ can only increase if signal $S$ becomes more informative. To make this statement formal, we consider a recommendation game with coupled signals $S,S'$. We give a definition that compares the two signals in terms of their ``state-relevant" information content, and state the ``monotonicity-in-information" lemma (proved in Appendix~\ref{app:coupled-signals}).

\OMIT{We say that signal $S$ is \emph{at least as informative} as $S'$ if the value
of $S$ determines the value of $S'$, \ie if $S' = g(S)$ for some function $g:\cX\to\cX'$.}

\begin{definition}\label{def:more-informative}
Consider a recommendation game with coupled signals $S,S'$, with resp. supports $\cX,\cX'$. We say that signal $S$ is \emph{at least as informative} as $S'$ if
\begin{align*}
     \Pr[\trueState=\theta \mid S=s,\, S'=s'] = \Pr[\trueState=\theta \mid S=s]
     \qquad \forall s\in \cX,\, s'\in\cX, \theta\in\Theta.
\end{align*}
An important special case is when $S$ \emph{determines} $S'$: $S'=g(S)$ for some $g:\cX\to\cX'$.
\end{definition}


\begin{lemma}[monotonicity-in-information]
\label{lm:coupled-exploit}
In a recommendation game with coupled signals $S,S'$, if signal $S$ is at least as  informative as $S'$ then $\REW^*[S]\geq \REW^*[S']$.
\end{lemma}


A recommendation policy $\pi$ for signal $S$ is also well-defined for signal $S'$ that determines $S$. Formally, $\pi$ induces a recommendation policy $\pi'$ that inputs $S'$, maps it to the corresponding value of $S$, and returns $\pi(S)$. Note that $\pi$ and $\pi'$ choose the same joint actions, and $\pi$ is BIC given $S$ if and only if $\pi'$ is BIC given $S'$. Henceforth we identify all such ``induced" policies $\pi'$ with $\pi$.

\begin{proof}[Proof of Claim~\ref{cl:restricted-subset}]
Consider the recommendation game that corresponds to the $t$-th round of Bayesian Exploration. Recall that the signal in this game is $S_t = (\omega,H_t)$, where $\omega$ is the internal random seed of policy $\pi$, and $H_t$ is the history before round $t$. Then $\pi^{(t)}$, the restriction of policy $\pi$ to round $t$, is a BIC recommendation policy in this game.

The expected reward of policy $\pi$ in round $t$ is the same as the expected reward of the restricted policy $\pi^{(t)}$, i.e.:
    $\REW_t(\pi) =\REW\rbr{\pi^{(t)}}$.
By optimality of $\REW^*[S_t]$ we have
    $\REW\rbr{\pi^{(t)}}\leq \REW^*[S_t]$.
Since policy $\pi$ is restricted to actions in $B$, signal $\signal(B)$ is at least as informative as signal $S_t$. So by Lemma~\ref{lm:coupled-exploit} we have
    $\REW^*[S_t] \leq \REW^*\sbr{ \signal(B) }$,
completing the proof.
\end{proof}

\subsection{Composition of subroutines}
\label{sec:warmup-subroutines}

We design iterative recommendation policies in a modular way, via ``subroutines" that comprise multiple rounds and accomplish a particular task. In particular, we need a formal framework to argue that our ``subroutines" are BIC if considered separately, and jointly form a BIC policy.

We model a ``subroutine" as an iterative recommendation policy that inputs the history of the previous rounds and chooses its own duration. More generally, we consider a common generalization of Bayesian Exploration and the recommendation game: the state $\trueState$ is drawn and the signal $S$ is observed exactly as in the recommendation game, and then the game proceeds over multiple rounds, exactly as in Bayesian Exploration. (Note that the recommendation game is simply a special case with time horizon $T=1$).

We focus on a version where the time horizon is infinite (and irrelevant). Instead, each iterative recommendation policy $\pi$ runs for $T_\pi$ rounds, where the number $T_\pi$, termed \emph{duration}, is chosen by the policy rather than given exogenously. The duration is chosen before the signal is observed, and thus can only depend on the utility structure and the signal structure. (This restriction is crucial, see Remark~\ref{rem:composition}.) The notion of BIC carries over word-by-word from Definition~\ref{def:BIC}.
Iterative recommendation policies in this model will also be called \emph{subroutines}.

From the computational point of view, the input to a subroutine consists of the utility structure, the signal structure, and the realization of the signal. The output of a subroutine is the tuple $(S,H)$, where $H$ is the \emph{history}: the chosen joint actions and the rewards/utilities for all rounds in the execution, or any function of $(S,H)$.

A recommendation game and a General Bayesian Exploration game with the same utility/signal structure are called \emph{associated}. As in Section~\ref{sec:coupled}, we fix the utility structure, but allow the signal $S$ to vary from one subroutine to another. A subroutine $\pi$ initialized with a particular feasible signal $s$ will be denoted $\pi(s)$.

Now we can define the \emph{composition} of subroutines. Consider two subroutines $\pi,\pi'$ with respective signals $S,S'$. Then the \emph{composition} of $\pi$ followed by $\pi'$, denoted $\pi\oplus\pi'$, is a subroutine with signal $S$ and duration $T_\pi+T_{\pi'}$. The first $T_\pi$ rounds of $\pi\oplus\pi'$ are controlled by $\pi$, and the subsequent $T_{\pi'}$ rounds are controlled by $\pi'$. In order for the composition to be well-defined, the signal for $\pi'$ should be expressed in terms of the output of $\pi$. Formally, we say that $\pi'$ is a \emph{valid sequel} for $\pi$ if the pair $(S,H)$ determines $S'$, where $H$ is the \emph{history} of $\pi$.

\begin{claim}\label{cl:warmup-composition}
Fix the utility structure. Consider BIC subroutines $\pi,\pi'$ such that $\pi'$ is a valid sequel for $\pi$. Then the composition $\pi\oplus\pi'$ is a BIC subroutine.
\end{claim}

\begin{proof}
\asedit{Consider the BIC condition for round $t$ of the composition $\pi\oplus\pi'$. If $t\leq T_{\pi}$ then this condition is equivalent for the BIC condition for round $t$ of subroutine $\pi$. Else, it is equivalent to the BIC condition for round $t-T_{\pi}$ of subroutine $\pi'$.}
\end{proof}

\begin{remark}\label{rem:composition}
\asedit{It is crucial for Claim~\ref{cl:warmup-composition} that the subroutine's durations are fixed in advance. If $T_{\pi}$ is state-dependent, observing a recommended action at a given round $t$ in the execution of $\pi\oplus\pi'$ may shed light on whether $\pi$ has terminated. This may leak something about the true state, breaking the incentives provided by $\pi$ and $\pi'$. Formally, the proof of Claim~\ref{cl:warmup-composition} breaks for the following reason: when one invokes the BIC property of subroutine $\pi$ (resp., $\pi'$), this property is now conditional on $t\leq T_{\pi}$ (resp., $t>T_{\pi}$), and therefore may no longer hold.}
\end{remark}

A sequence of subroutines $(\pi,\pi',\pi'',\ldots)$ is called \emph{valid} if the every next subroutine is a valid sequel for the previous one. It is easy to see that the composition is associative:
    $\pi\oplus(\pi'\oplus\pi'') = (\pi\oplus\pi')\oplus\pi''$,
so it is uniquely defined by the sequence, and can be denoted $\pi\oplus\pi'\oplus\pi''$. Thus, a BIC iterative recommendation policy can be presented as a composition of a valid sequence of subroutines, where the first subroutine inputs an empty signal, and all durations sum up to $T$.

\OMIT{As a shorthand, the \emph{completion} of subroutine $\pi$ with a recommendation policy $\pi^*$ (assuming $\pi^*$ is a valid sequel for $\pi$) is the composition of $\pi$ followed by infinitely many copies of $\pi^*$.}


\subsection{Proof of Lemma~\ref{lm:warmup}}

Let $\pi$ be the policy from the lemma statement, and let $\sigma$ be the subroutine of duration $T_\pi$ which coincides with $\pi$ on the first $T_\pi$ rounds. (In other words, $\sigma$ is a version of policy $\pi$ that is ``truncated" after round $T_\pi$.) Let $\pi^*$ be an optimal recommendation policy for signal $S=\signal(\cA)$. Note that
    $\REW^*[S]\geq \opt(\PiBIC)$
by Claim~\ref{cl:restricted-subset}.
Let $\pi'$ be an iterative recommendation policy defined as the composition of $\sigma$ followed by $T-T_\sigma$ copies of $\pi^*$. Note that $\pi'$ is BIC by Claim~\ref{cl:warmup-composition}, and receives expected reward $\REW^*[\cA]$ in each round after $T_\pi$ by construction.



\section{Bayesian Exploration with deterministic utilities}
\label{sec:det}

In this section, we focus on deterministic utilities, and prove
Theorem~\ref{thm:model-det}. In particular, we will construct a BIC
iterative recommendation policy $\pi$ whose expected reward
$\REW(\pi)$ satisfies \eqref{eq:thm:model-det}.

Rather than comparing $\pi$ directly to $\OPT(\PiBIC)$, we will use an intermediate benchmark of the form $\REW^*\sbr{ \signal(\explorableD)}$, for some subset $\explorableD\subset \cA$ of joint actions that depends on the realized state $\trueState$. While we used $\explorableD=\cA$ to prove Lemma~\ref{lm:warmup}, this is not a ``fair" intermediate benchmark, because there may be some joint actions that cannot be explored by any BIC policy. Instead, we will only consider joint actions that can be explored.

\begin{definition}[explorability]\label{def:explorability}
A joint action $a\in\cA$ is called \emph{eventually-explorable} given a state $\theta\in\Theta$ if
    $\Pr[\pi^t=a\mid\trueState=\theta]>0$
for some BIC iterative recommendation policy $\pi$ and some round $t\in \NN$. The set of all such joint actions is denoted $\explorableD[\theta]$.
\end{definition}

Thus, $\explorableD$ is the set of all joint actions that can be explored by some BIC policy in some round, given the realized state $\trueState$. Note that it is a random variable whose realization is determined by $\trueState$. Using $\REW^*[\explorableD]$ as an intermediate benchmark suffices because, by Claim~\ref{cl:restricted-subset} we have
\begin{align}\label{eq:explorableD}
\REW^*\sbr{\explorableD}\geq \opt(\PiBIC).
\end{align}

\begin{remark}\label{rem:explorable}
To see that $\explorableD$ may depend on $\trueState$, consider the following simple example. Assume a basic scenario with a   single agent per round, two possible actions $a_1,a_2$, and deterministic rewards (denoted $R_1$ and $R_2$). The prior for each action chooses reward independently and uniformly from two possible values: $\{\tfrac12, 1\}$ for $R_1$, and $\{0,1\}$ for $R_2$. Then $a_2$ can be explored if and only if $R_1=\tfrac12$. (Why? In the first round the algorithm must recommend $a_1$. In every subsequent round before $a_2$ is chosen, the algorithm's policy is a randomized mapping $\pi:R_1\to \{a_1,a_2\}$. If $\pi$ puts a positive probability on $a_2$ after seeing $R_1=1$, then $\pi$ is not BIC: indeed, an agent would prefer to deviate to $a_1$ when recommended $a_2$. On the other hand, a policy $\pi$ which deterministically chooses $a_1$ when $R_1=1$, and $a_2$ otherwise, is BIC.)
\end{remark}

We construct a BIC iterative recommendation policy which explores all of $\explorableD$. Note that the existence of such policy does not immediately follow from Definition~\ref{def:explorability}, because the definition only guarantees a BIC policy separately for each $(\theta,a)$ pair, whereas we need one BIC policy which ``works" for the specific (and unknown) realized state $\trueState$, and all ``matching" joint actions $a$.

\begin{definition}
A subroutine is called \emph{maximally-exploring} if it is a BIC subroutine that inputs an empty signal and explores all joint actions in $\explorableD$ by the time it stops.
\end{definition}

A maximally-exploring subroutine can be followed by \emph{exploitation}, using the corresponding optimal recommendation policy. The resulting  BIC iterative recommendation policy achieves the regret bound claimed in Theorem~\ref{thm:model-det}.

\begin{lemma}\label{lm:exploring-reward}
Let $\sigma$ be a maximally-exploring subroutine of duration $T_\sigma$. Let $\pi^*$ be an optimal recommendation policy for signal $\signal(\explorableD)$. Let $\pi$ be the composition of subroutine $\sigma$ followed by $T-T_\sigma$ copies of $\pi^*$. Then $\pi$ is a BIC iterative recommendation policy that satisfies
\begin{align}
\REW(\pi) \geq (T-T_\sigma)\;\opt(\PiBIC).
\end{align}
\end{lemma}

\noindent The lemma easily follows from the machinery developed in Section~\ref{sec:warmup}. Namely, we use the notion and existence of $\explorableD$-optimal policy (see Corollary~\ref{cor:LP-rep}), the ``monotonicity-in-information" analysis which guarantees~\eqref{eq:explorableD}, and the ``composition of subroutines" analysis which guarantees that $\pi$ is BIC (via Claim~\ref{cl:warmup-composition}).

The existence of a maximally-exploring subroutine also allows to characterize the benchmark $\opt(\PiBIC)$ as the maximal per-round expected reward of a BIC policy that has access to the utilities of all explorable joint actions.

\begin{theorem}\label{thm:deterministic-benchmark-characterization}
For Bayesian Exploration with deterministic utilities,
$\opt(\PiBIC) = \REW^*[\explorableD]$.
\end{theorem}

\begin{corollary}\label{cor:deterministic-benchmark-characterization}
Consider Bayesian Exploration with deterministic utilities. Assume that there is only a single agent per round, the principal's reward coincides with the agent's utility, and all actions are explorable. Then the benchmark is the expected reward of the best action:
    \[ \opt(\PiBIC) = \E_{\trueState\sim\prior}\sbr{ \max_{a\in\cA} f(a,\trueState) }. \]
\end{corollary}

The rest of this section is organized as follows. In Section~\ref{sec:det-RecGame-explore} we develop the theory of ``explorability" in a recommendation game. In Section~\ref{sec:det-maxExplore} we use this theory to define a natural BIC subroutine and prove that it is maximally-exploring (and has the desired per-round computation time). Finally, in Section~\ref{sec:det-wrapup} we put it all together to prove Theorem~\ref{thm:model-det} and Theorem~\ref{thm:deterministic-benchmark-characterization}.

\subsection{Explorability in the recommendation game}
\label{sec:det-RecGame-explore}

In this subsection we investigate which joint actions can be explored in the recommendation game. We adopt a very permissive definition of ``explorability", study some of its properties, and design a subroutine  which explores all such joint actions. Throughout, we consider a recommendation game with signal $S$ whose support is $\cX$.

\begin{definition}\label{def:signal-explorable}
Consider a recommendation game with signal $S$. A joint action $a$ is called \emph{signal-explorable}, for a given feasible signal $s\in \cX$, if there exists a BIC recommendation policy $\pi = \pi^{a,s}$ such that $\Pr[\pi(s)=a]>0$. The set of all such joint actions is denoted $\Ex_s[S]$. The \emph{signal-explorable set} is defined as $\Ex[S] = \Ex_S[S]$.
\end{definition}

Note that $\Ex_s[S]$ is a fixed subset of joint actions, determined by the feasible signal $s\in\cX$, whereas $\Ex[S]$ is a random variable (in fact, a $2^\cA$-valued signal) whose realization is determined by the realization of signal $S$.

We show that the signal-explorable set $\Ex[S]$ can only increase if signal $S$ becomes more informative. (The proof is deferred to Appendix~\ref{app:coupled-signals}.)

\begin{lemma}[monotonicity-in-information for explorability]
\label{lm:coupled-explore}
Consider a recommendation game with coupled signals $S,S'$. If signal $S$ is at least as informative as $S'$ then $\Ex[S] \supset \Ex[S']$.
\end{lemma}

\noindent Note that $\Ex[S]$ and $\Ex[S']$ are set-valued random variables, and the claim asserts that one random set is always contained in the other. (In general, if $X$ and $Y$ are random variables, we will write $X\in Y$ and $X\subset Y$ to mean that the corresponding event holds for all realizations of randomness.)

\xhdr{A recommendation policy with maximal support.}
Observe that policies $\pi^{a,s}$ in Definition~\ref{def:signal-explorable} can be replaced with a single BIC recommendation policy $\pimax$ such that $\Ex_s[S] = \support\rbr{\pimax(s)}$ for each feasible signal  $s\in\cX$. We will call such $\pimax$ a \emph{max-support} policy for signal $S$. For example, we can set
\begin{align}\label{eq:cl:max-support-policy}
\textstyle \pimax = \frac{1}{|\cX|} \,\sum_{s\in \cX} \frac{1}{|\;\Ex_s[S]\;|} \;
    \sum_{a\in \Ex_s[S]} \pi^{a,s}.
\end{align}
This policy is BIC as a convex combination of BIC policies, and max-support by design.

We compute a max-support policy as follows. For each joint action $a\in \cA$ and each feasible signal $s\in \cX$, we solve the following LP:
\begin{align}\label{eq:LP-Xas}
 \max_{x \in \bic[S]}  \eta_{a,s} \quad
 \mbox{subject to } \quad  x_{a, s} \geq \eta_{a,s}.
\end{align}
This LP always has a solution $\eta_{a,s}\geq 0$, because BIC policies exist and any BIC policy gives a solution $x\in\bic[S]$ with $\eta_{a,s}= 0$. Further, $\eta_{a,s}>0$ if and only if $a\in \Ex_s[S]$, in which case the solution $x$ is the LP-representation of a policy $\pi^{a,s}$ such that $\Pr[\pi^{a,s}(s)=a]>0$. Then the max-support policy $\pimax$ is computed via \eqref{eq:cl:max-support-policy}. The computational procedure is given in Algorithm~\ref{alg:MaxSupport}. Note that $\maxSupport$ computes subsets $\cA_s=\Ex_s[S]$, and the output $x$ is the LP-representation of the policy given by \eqref{eq:cl:max-support-policy}.

\begin{algorithm}[htb]
\begin{algorithmic}
\STATE {\bf Input:} the utility structure and the signal structure.
\STATE {\bf Output:} an LP-representation of a max-support policy.
\vspace{2mm}
\STATE For each joint action $a\in \cA$ and feasible signal $s\in \cX$:
\STATE \myTab solve the linear program~\eqref{eq:LP-Xas};
    let $x^{a,s}$ be a solution with objective value $\eta_{a,s}$.
\STATE Let $\cA_s = \{a\in \cA:\; \eta_{a,s}>0\}$ for each $s\in\cX$.
\STATE Output
    $x = \frac{1}{|\cX|} \,\sum_{s\in \cX} \frac{1}{|\cA_s|} \; \sum_{a\in \cA_s} x^{a,s}$.
\end{algorithmic}
\caption{$\maxSupport$: computes a max-support policy.}
\label{alg:MaxSupport}
\end{algorithm}

The ``quality" of a max-support policy $\pimax$ is, for our purposes, expressed by the minimal probability of choosing a joint action in its support:
\begin{align}\label{eq:pmin-pi}
\pmin(\pimax)
    := \quad\min_{s\in \cX,\; a\in \support\rbr{\pimax(s)}}\quad \Pr\sbr{\pimax(s)=a}.
\end{align}
We relate this quantity to the minimax probability in Definition~\ref{def:signal-explorable}:
\begin{align}\label{eq:det-pmin}
\psig[S]
    := \quad\min_{s\in \cX,\; a\in \Ex_s[S]}\quad
            \max_{x\in\bic[S]}\quad x_{a,s}.
\end{align}
It is easy to see that $\maxSupport$ constructs policy $\pimax$ with
    $\pmin(\pimax)\geq \frac{1}{|\cA|\cdot |\cX|}\; \psig[S]$.
To summarize, we have proved the following:

\begin{claim}\label{cl:max-support-policy}
There is a max-support policy $\pimax$ with
    $\pmin(\pimax)\geq \frac{1}{|\cA|\cdot |\cX|}\; \psig[S]$.
It can be computed by $\maxSupport$ in time polynomial in
$|\cA|$, $|\cX|$, and $|\Theta|$.
\end{claim}

\xhdr{Maximal exploration given a signal.}
Let us design a subroutine $\sigma$ which explores all signal-explorable joint actions. More specifically, given a recommendation game with signal $S$, we are looking for a subroutine $\sigma$ in the associated Generalized Bayesian Exploration game which explores all joint actions in $\Ex[S]$. Such subroutine is called \emph{maximally-exploring} for signal $S$.

\newcommand{\xmax}{x^{\max}}

We start with a max-support policy $\pimax$ returned by $\maxSupport$. Let $\xmax$ be its LP-representation, so that $\xmax_{a,s} := \Pr[\pimax(s)=a]$. Given a particular signal realization $s\in\cX$, we proceed as follows. We compute the signal-explorable set $\Ex_s[S]$ as the support of $\pimax(s)$. For each joint action $a\in \Ex_s[S]$, we choose the \emph{dedicated round} $\tau_s(a)$ when this action is chosen. The dedicated rounds are chosen uniformly at random, in the sense that an injective function
    $\tau_s:\Ex_s[S] \to [T_\sigma]$
is chosen uniformly at random among all such functions (where $T_\sigma$ is the duration of the subroutine). To guarantee that the subroutine is BIC, we ensure that
\begin{align}\label{eq:max-exploring-subroutine-prop}
\Pr\sbr{\sigma^t(s) =a} =\xmax_{a,s}
    \quad\text{for each round $t$ and joint action $a\in\cA$}.
\end{align}
To this end, for each non-dedicated round the joint action is chosen independently from a fixed distribution $\cD_s$ that is constructed to imply \eqref{eq:max-exploring-subroutine-prop}. Specifically, we write $N_s = |\Ex_s[S]|$ and define:
\begin{align}\label{eq:remainder-distribution}
    \cD_s(a) = \frac{ \xmax_{a,s}-1/T_\sigma}{T_\sigma-N_s}\quad \forall a\in\Ex_s[S].
\end{align}
It is easy to see that $\cD_s$ is a distribution, provided that the duration is large enough. Specifically, it suffices to set
    $T_\sigma = \max\left( 1+N_s ,\; \cel{1/\pmin(\pimax)}\right)$.
Then for each round $t$ and action $a\in \Ex_s[S]$,
\begin{align*}
\Pr\sbr{ \sigma^t(s) =a}
    = 1/T_\sigma + \cD_s(a)\cdot (T_\sigma - N_s)
    =\xmax_{a,s},
\end{align*}
where the $1/T_\sigma$ is the probability that round $t$ is dedicated for action $a$. This completes the description of the subroutine, and the proof that it is BIC. The computational procedure for this subroutine is summarized in Algorithm~\ref{alg:MaxEx}.

\begin{algorithm}[ht!]
\begin{algorithmic}
\STATE {\bf Input:} utility structure $\cU$, signal structure $\sigStruc$, and signal realization $s\in\cX$.
\vspace{2mm}

\COMMENT{compute the parameters}
\STATE $x \leftarrow \maxSupport(\cU,\sigStruc)$
    \COMMENT{LP-representation of a max-support policy $\pimax$}
\STATE $B \leftarrow \{ a:\in \cA:\; x_{a,s}>0\} $
    \COMMENT{signal-explorable set $\Ex_s[S]$}
\STATE $\pmin \leftarrow \min(x_{a,s'}:\; a\in B,\; s'\in\cX)$
    \COMMENT{computes $\pmin(\pimax)$}
\STATE $T\leftarrow \max(1+|B|,\; \cel{1/\pmin})$
    \COMMENT{the duration of the subroutine}
\STATE Pick injective function $\tau:B\to [T]$ u.a.r. from all such functions.
\STATE
    $\cD(a) \leftarrow \frac{x_{a,s}-1/T}{T-|B|}\quad \forall a\in B$.
    \COMMENT{distribution \eqref{eq:remainder-distribution} for non-dedicated rounds}
\vspace{2mm}

\COMMENT{issue the recommendations}
\FOR{rounds $t=1 \ldots T$}
\STATE {\bf if}
        $t=\tau(a)$ for some $a\in B$
    {\bf then}
        $a_t\leftarrow a$
    {\bf else}
         choose $a_t$ from distribution $\cD$
\STATE Recommend action $a_t$, observe the corresponding utilities
\ENDFOR

\vspace{2mm}
\STATE {\bf Output:} signal $\signal(\{a_1 \LDOTS a_T\})$.
\end{algorithmic}
\caption{Subroutine $\MaxEx$: maximal exploration given signal $S$.}
\label{alg:MaxEx}
\end{algorithm}

%

Our discussion of $\MaxEx$ can be summarized as the following claim:

\begin{claim}\label{cl:det-MaxExplore-signal}
$\MaxEx$ is a maximally-exploring subroutine for signal $S$. Its duration is at most
    $|\cA|\cdot |\cX|\;/\;\psig[S]$
rounds. The running time is polynomial in $|\cA|\cdot|\cX|$ in pre-processing, and linear in $|\cA|$ per each round.
\end{claim}

\subsection{A maximally-exploring subroutine}
\label{sec:det-maxExplore}

We come back to Bayesian Exploration, and strive to explore as many joint actions as possible. We define a natural BIC subroutine for this goal, and prove that it is indeed maximally-exploring.

Our BIC subroutine is based on the following intuition. Initially, one can explore some joint actions via ``maximal exploration" for an empty signal $S_1=\bot$. This gives some additional observations, so one can now run "maximal exploration" for a new signal $S_2$ which comprises these new observations. This in turn provides some new observations, and so forth. We stop after $|\cA|$ iterations, which (as we prove) suffices to guarantee that no further progress can be made.

Formally, the subroutine proceeds in phases
    $\ell = 1,2,3 \LDOTS |\cA|$.
In each phase, we start with the ``current" signal $S_\ell$, and perform ``maximal exploration" given this signal by calling $\MaxEx$. This call computes the ``next" signal, defined recursively as 
\begin{align}\label{eq:det-signal}
    S_{\ell+1} = \signal\rbr{\Ex[S_\ell]}.
\end{align}
After the last phase, we output the latest signal
    $S_{|\cA|+1}$
which (as we prove) encompasses all observations received throughout the subroutine.
Each call to $\MaxEx$ must be parameterized with the signal structure for the corresponding signal $S_\ell$. Since the signal is uniquely determined by the realized state $\trueState$, the signal structure is completely specified by the tuple
    $\sigStruc_{\ell} = \rbr{S_\ell(\theta):\; \theta\in\Theta}$,
where $S_\ell(\theta)$ is the value of the signal that corresponds to realized state $\trueState=\theta$. Henceforth we identify $\sigStruc_{\ell}$ with the signal structure.
All signal structures are computed before the repeated game starts.
The pseudocode is summarized in Algorithm~\ref{alg:induct}.

\begin{algorithm}[ht!]
\begin{algorithmic}
\STATE {\bf Input:} the utility structure $\cU$.
\vspace{2mm}

\STATE {\bf Initialize}: $S_1 =\sigStruc_1 = \bot$.
    \COMMENT{``phase-$1$ signal" is empty}

\vspace{2mm}
\COMMENT{compute signal structure $\sigStruc_{\ell}$ for each signal $S_{\ell}$, 
    as per \eqref{eq:det-signal}}
    \ascomment{new loop}
\FOR{each phase $\ell=1,2 \LDOTS |\cA|$}
\STATE $x \leftarrow \maxSupport(\cU,\sigStruc_\ell)$
    \COMMENT{LP-representation of a max-support policy}
\asedit{\FOR{each state $\theta\in\Theta$}
\STATE  $s\leftarrow S_\ell(\theta)$
    \COMMENT{value of signal $S_\ell$ if $\trueState=\theta$}.
\STATE
    $S_{\ell+1}(\theta)\leftarrow \signal\rbr{\{ a\in\cA: x_{a,s}>0\}}$
    \COMMENT{if $\trueState=\theta$ then
        $\Ex[S_\ell] = \cbr{ a\in\cA: x_{a,s}>0 }$.}
\ENDFOR}
\STATE $\sigStruc_{\ell+1} \leftarrow \rbr{S_{\ell+1}(\theta):\; \theta\in\Theta}$.
    \COMMENT{signal structure for $S_{\ell+1}$}
\ENDFOR

\vspace{2mm}
\STATE \COMMENT{play the repeated game}
\FOR{each phase $\ell=1,2 \LDOTS |\cA|$}
\STATE $S_{\ell+1} \leftarrow \MaxEx(\cU,\sigStruc_{\ell},S_\ell)$
            \COMMENT{compute signal $S_{\ell+1}$, as per \eqref{eq:det-signal}}
\ENDFOR

\vspace{2mm}
\STATE {\bf Output:} signal $S_{|\cA|+1}$.
\end{algorithmic}
\caption{Subroutine $\IndMax$: maximal exploration.}
\label{alg:induct}
\end{algorithm}

Let us summarize the immediate properties of $\IndMax$. It is BIC as a composition of BIC subroutines by Claim~\ref{cl:warmup-composition}. Since the support of each signal $S_\ell$ is of size at most $|\Theta|$,
the per-round running time is polynomial in $|\cA|\cdot |\Theta|$. 
This is because the running time of $\maxSupport$ and the per-round running time of $\MaxEx$ are polynomial in $|\cA|\cdot |\Theta|$, by Claims~\ref{cl:max-support-policy} and~\ref{cl:det-MaxExplore-signal}. The call to $\MaxEx$ in each phase $\ell$ completes in
    $\poly\rbr{|\cA|\cdot |\cX|\;/\;\psig[S_\ell]}$ 
rounds by Claim~\ref{cl:det-MaxExplore-signal}. Thus:

\begin{claim}\label{cl:det-comput}
Subroutine $\IndMax$ is BIC; its per-round running time is polynomial in $|\cA|\cdot |\Theta|$;
\asedit{its duration is at most $|\cA|^2\cdot |\Theta|/\psig^*$ rounds, where
    $\psig^* = \min_{\ell\leq |\cA|}\;\psig[S_\ell]$.}
\end{claim}

\begin{remark}\label{rem:prior-dep}
\asedit{As per Claim~\ref{cl:det-comput}, 
    $\psig^* = \min_{\ell\leq |\cA|}\;\psig[S_\ell]$
is the key parameter which determines the number of rounds in $\IndMax$. It is determined by the utility structure: recall that $\psig[S]$ is defined in \eqref{eq:det-pmin} for any signal $S$, and signals $S_{\ell}$, $\ell\leq |\cA|$ are defined recursively in \eqref{eq:det-signal}.}
\end{remark}

Let $B_{\ell+1}$ be the set of all joint actions explored during phase $\ell$. For the sake of the argument, let us define $B_1 = \emptyset$, and extend the definition of sets $B_\ell$ and signals $S_\ell$ to phases $\ell = 1,2,3, \ldots\;$ (i.e., without an upper bound on the number of phases). By construction,
    $B_{\ell+1}=\Ex[S_\ell]$
is a random variable whose realization is determined by the realized state $\trueState$,
and
    $S_{\ell} = \signal(B_\ell)$.

We show that $B_{\ell+1}$ is the set of all joint actions explored during the first $\ell$ phases, and that stopping after $|\cA|$ phases is without loss of generality:

\begin{claim} \label{cl:det-Bl}
Sets $(B_\ell:\;\ell\in\NN)$ are non-decreasing in $\ell$, and identical for $\ell\geq |\cA|+1$.
\end{claim}

\begin{proof}
To prove that sets $(B_\ell:\;\ell\in\NN)$ are non-decreasing in $\ell$, we use induction on $\ell$ and Lemma~\ref{lm:coupled-explore}(a) on ``monotonicity-in-information". Now, a strictly increasing sequence of subsets of $\cA$, starting from an empty set, cannot have more than $|\cA|+1$ elements. It follows that $B_\ell = B_{\ell+1}$ for some phase $\ell\leq |\cA|+1$. By definition of $B_\ell$, the sets $(B_{\ell'}:\;\ell'\geq\ell)$ are identical.
\end{proof}

Depending on the state $\trueState$, some of the later phases may be redundant in terms of exploration, in the sense that
    $B_\ell = B_{\ell+1} = \ldots = B_{|\cA|+1}$
starting from some phase $\ell$. Nevertheless, we use a fixed number of phases so as to ensure that the duration is fixed in advance. (This is required by the definition of a subroutine so as to ensure composability, as per Claim~\ref{cl:warmup-composition}.)

\begin{lemma}\label{lm:det-IndMax-exploring}
Subroutine $\IndMax$ is maximally-exploring.
\end{lemma}

\begin{proof}
We consider an arbitrary BIC iterative recommendation policy $\pi$, and prove that all joint actions explored by $\pi$ are also explored by $\IndMax$.

Denote the internal random seed in $\pi$ by $\omega$. Without loss of generality, $\omega$ is chosen once, before the first round, and persists throughout. Let $A_t$ be the set of joint actions explored by $\pi$ in the first $t-1$ rounds. It is a random variable whose realization is determined by the random seed $\omega$ and realized state $\trueState$. Let us represent policy $\pi$ as a sequence of $(\pi^{(t)}:t\in [T])$, where each $\pi^{(t)}$ is a BIC recommendation policy which inputs signal $S^*_t$ which consists of the random seed and the observations so far:
    $S^*_t = (\omega,\signal(A_t))$,
and is deterministic given that signal. In each round $t$, the recommended joint action is chosen by policy $\pi^{(t)}$; it is denoted $\pi^{(t)}(S^*_t)$.

We will prove the following claim using induction on round $t$:
\begin{align}\label{eq:pf:lm:det-IndMax-exploring}
\text{for each round $t$, there exists phase $\ell\in\NN$ such that
    $\pi^{(t)}(S^*_t) \in B_\ell$.}
\end{align}

For the induction base, consider round $t=1$. Then $A_1=\emptyset$, so the signal is simply
    $S^*_1 = (\omega,\bot)$.
Note that the empty signal $\bot$ is at least as informative as the signal $S^*_1$. Therefore:
\begin{align*}
\pi^{(1)}(S^*_1)
    &\in \Ex\sbr{ S^*_1 }     &\text{(by definition of $\Ex\sbr{S^*_1}$)} \\
    &\subset \Ex[\bot]  &\text{(by Lemma~\ref{lm:coupled-explore})}\\
    &=B_1            &\text{(by definition of $B_1$)}.
\end{align*}

For the induction step, assume that \eqref{eq:pf:lm:det-IndMax-exploring} holds for all rounds $t<t_0$, for some $t_0$. Then $A_{t_0} \subset B_\ell$ for some phase $\ell$, so
signal
    $S_\ell = \signal(B_\ell)$
for phase $\ell$ of $\IndMax$ is at least as informative as signal
     $S^*_{t_0} = (\omega, \signal(A_{t_0}))$
for round $t_0$ of policy $\pi$. Therefore:
\begin{align*}
\pi^{(t_0)}(S^*_{t_0})
    &\in \Ex\sbr{S^*_{t_0}}     &\text{(by definition of $\Ex\sbr{S^*_{t_0}}$)} \\
    &\subset \Ex[S_\ell]    &\text{(by Lemma~\ref{lm:coupled-explore})}\\
    &=B_{\ell+1}                &\text{(by definition of $B_\ell$)}.
\end{align*}
This completes the proof of \eqref{eq:pf:lm:det-IndMax-exploring}. Therefore policy $\pi$ can only explore joint actions in $\cup_{\ell\in\NN} B_\ell$, and by
Claim~\ref{cl:det-Bl} this is just $B_{|\cA+1|}$, the set of joint actions explored by $\IndMax$.
\end{proof}

\subsection{Putting this all together: proof of Theorem~\ref{thm:model-det} and Theorem~\ref{thm:deterministic-benchmark-characterization}}
\label{sec:det-wrapup}

We use the maximally-exploring subroutine $\IndMax$ from Section~\ref{sec:det-maxExplore} (let $T_0$ denote its duration), and an optimal recommendation policy for signal $S=\signal\rbr{\explorableD}$, denote it $\pi^*$. Let $\pi$ be the composition of subroutine $\IndMax$ followed by $T-T_0$ copies of $\pi^*$. By Lemma~\ref{lm:exploring-reward}, this policy has the reward guarantee claimed in Theorem~\ref{thm:model-det}.

Recall that $\REW(\pi^*) = \REW^*\sbr{\explorableD}$, which is at least $\opt(\PiBIC)$ by~\eqref{eq:explorableD}. It follows that the supremum in the definition of $\opt(\PiBIC)$ is attained by $\pi\in \PiBIC$. Which implies Theorem~\ref{thm:deterministic-benchmark-characterization}.

To complete the proof of Theorem~\ref{thm:model-det}, let us argue about the computational implementation of $\pi$. For brevity, let us say \emph{polytime-computable} to mean ``computable in time polynomial in $|\cA|\cdot|\Theta|$". We need to prove that each round of $\pi$ is polytime-computable. Each round of $\IndMax$ is polytime-computable by Claim~\ref{cl:det-comput}, and each round of $\pi^*$ is polytime-computable by definition, so it remains to prove that a suitable $\pi^*$ is polytime-computable.

By Corollary~\ref{cor:LP-rep}, policy $\pi^*$ is poly-time computable given the signal structure for signal $S$ (because the signal is determined by the realized state $\trueState$, and therefore has support of size at most $|\Theta|$). Recall that $S = S_{|\cA+1|}$ is the signal computed in the last phase of $\IndMax$. So the corresponding signal structure is poly-time computable using a version of $\IndMax$ without the calls to $\MaxEx$. This completes the proof of Theorem~\ref{thm:model-det}.

\newpage
\section{Bayesian Exploration with stochastic utilities}
\label{sec:noisy}

We turn our attention to stochastic utilities. As we pointed out in Section~\ref{sec:model-results}, we compete against $\delta$-BIC iterative recommendation policies, for a given $\delta>0$.%
\footnote{Recall that $\delta$-BIC policies satisfy a stronger version of Definition~\ref{def:BIC} in which right-hand side of \eqref{eq:def:BIC} is $\delta$.} We construct a BIC policy whose time-averaged expected reward is close to $\OPT(\PiBIC[\delta])$, where $\PiBIC[\delta]$ is the class of all $\delta$-BIC policies.

In fact, we achieve a stronger result. We consider the~\emph{deterministic instance}: a version of the original problem instance with deterministic utilities. For the deterministic instance, let $\PiBICdet[\delta]$ be the class of all $\delta$-BIC policies, and let $\optdet(\cdot)$ be the value of $\opt(\cdot)$. We will compete against
    $\optdet(\PiBICdet[\delta])$.
This is indeed a stronger benchmark:

\begin{lemma}\label{lm:model-benchmarks}
$\optdet(\PiBICdet[\delta])\geq \opt(\PiBIC[\delta])$ for all $\delta>0$.
\end{lemma}
\begin{proof}
Fix $\delta>0$. Given an arbitrary $\delta$-BIC iterative recommendation policy $\pi$ for the original instance, one can construct an iterative recommendation policy $\sigma$ such that for the deterministic instance this policy is $\delta$-BIC and achieves the same per-round expected rewards:
    \[ \E\sbr{ f(\sigma^t,\trueState)} = \E\sbr{ f(\pi^t,\trueState) }
        \quad\text{for each round $t$}. \]
The construction is very simple: in each round $t$ of the deterministic instance, $\pi'$ recommends the joint action $a=a_t$ selected by policy $\pi$, observes the utility tuple $\cU_{(a,\trueState)}$, draws a ``fake" vector of realized utilities from the corresponding distribution $\cD_{(a,\trueState)}$, and feeds this vector to $\pi$.

To complete the proof, recall that $\opt(\cdot)$ in \eqref{eq:model-benchmark} is defined via $\sup$. For any $\eps>0$ there exists a $\delta$-BIC iterative recommendation policy for the original instance which comes $\eps$-close, in the sense that
    $\E[f(\pi^t,\trueState)] > \opt(\PiBIC[\delta]) - \eps$
for some round $t$. Taking the corresponding policy $\sigma$ as above,
\[ \optdet(\PiBICdet[\delta])
   \geq \E\sbr{ f(\sigma^t,\trueState)}
   =\E\sbr{f(\pi^t,\trueState)}
   > \opt(\PiBIC[\delta]) - \eps. \qedhere \]
\end{proof}

%

\asmargincomment{Removed the thm stt in terms of $\PiBICdet[\delta]$.}
Our policy (and regret bounds) depend on two parameters determined by the utility structure: a version of $\psig^*$ from Remark~\ref{rem:prior-dep}, and the \emph{separation parameter} defined below:

\begin{definition}
The \emph{separation parameter} is the smallest number $\myGap > 0$ such that for any states $\theta, \theta' \in \Theta$, any agent $i\in[n]$, and any joint action $a\in \cA$ we have
\begin{align*}
u_i(a, \theta) \neq u_i(a, \theta') \;\Rightarrow\;
 |u_i(a, \theta) - u_i(a, \theta')| \geq \myGap.
\end{align*}
\end{definition}

\asedit{In words: fixing a joint action and going from one state to another, the utility of any agent changes by at least $\myGap$, if it changes at all. Less formally, if the utilities change when the state changes, then the change must be sufficiently large. The separation parameter is superficially similar to ``gap" in stochastic bandits (the difference in expected rewards between the best and the second-best arms), but different in that $\myGap$ refers to the difference between states, not arms.}

The main steps in our solution are as follows.
First, we extend some basic tools and concepts of BIC
  recommendation policies to $\delta$-BIC recommendation policies.
Second, we give a useful subroutine that approximates the
  expected utilities using multiple samples of the stochastic
  utilities. We show that any $\delta$-BIC recommendation policy for the deterministic
  instance remains BIC when it inputs approximate utilities instead of the expected utilities,
as long as the approximation is sufficiently good.
Third, we present a maximally-exploring subroutine, analogous to the one in Section~\ref{sec:det}, and use it to construct a BIC recommendation policy that achieves
  logarithmic regret compared to the class of all $\delta$-BIC policies.

\subsection{Preliminaries: from BIC to $\delta$-BIC}

We extend some of the machinery developed in the previous sections to $\delta$-BIC policies. In the interest of space, we sometimes refer to the said machinery rather than spell out the full details.

We start with the notions of ``eventually-explorable joint action" and ``maximally-exploring subroutine". Note that the former only considers the deterministic instance.

\begin{definition}[$\delta$-explorability]
Consider the deterministic instance and fix $\delta\geq 0$. A joint action $a\in \cA$ is called \emph{$\delta$-eventually-explorable} given a state
    $\theta \in \Theta$ if $\Pr[\pi^t
  = a\mid \theta] > 0$ for some $\delta$-BIC iterative recommendation
policy $\pi$ and some
round $t\in \NN$. The set of all such joint actions is denoted
$\cA_\theta^{\delta}$.
\end{definition}

\begin{definition}
A subroutine is called $\delta$-maximally-exploring, for a given $\delta>0$, if it is a BIC subroutine that inputs an empty signal and explores all joint actions in
$\cA_{\trueState}^{\delta}$ by the time it stops.
\end{definition}

Let us turn to the recommendation game. For the recommendation game, the distinction between deterministic and stochastic utilities is irrelevant, as far as this paper is concerned, because only the \emph{expected} utilities matter.

\begin{definition}\label{def:sex}
Consider a recommendation game with signal $S$, and fix $\delta\geq 0$. Policy $\pi$ is called \emph{$\delta$-BIC} if it satisfies a version of Definition~\ref{def:rec-BIC} in which the right-hand side of~\eqref{eq:def:rec-BIC} is $\delta$ rather than $0$. Further:
\begin{OneLiners}
\item A joint action $a$ is called $\delta$-\emph{signal-explorable}, for a particular feasible signal $s\in\cX$,
if there exists a $\delta$-BIC recommendation policy $\pi^{a, s}$ such that
    $\Pr[\pi^{a, s}(s) = a] > 0$.
The set of all such joint actions is denoted $\SEx_s[S]$. The \emph{$\delta$-signal-explorable set} is defined as a random set
    $\SEx[S] := \SEx_S[S]$.
\item A $\delta$-BIC policy $\pi$ is called \emph{$\delta$-max-support} if
    $\support(\pi(s)) = \SEx_s[S]$
for each signal $s\in\cX$.
\end{OneLiners}
\end{definition}

The LP-representation of a $\delta$-BIC recommendation policy satisfies a version of the LP
from Section~\ref{sec:rec-game-props} where the first constraint expresses $\delta$-BIC condition rather than BIC. Specifically, the right-hand side in the first constraint should be changed from $0$ to $\delta\cdot \Pr[\pi_i(S) = a_i] $. In terms of the LP variables, the right-hand side becomes
$\delta\cdot \sum_{a_{-i}\in \cA_{-i},\; s\in\cI,\; \theta\in \Theta}\;
    \prior(\theta)
    \cdot\joint(s\mid\theta)
    \cdot x_{a, s}.
$
The feasible region of the modified LP will be denoted $\bic_\delta[S]$.

\begin{claim}\label{cl:LP-repN}
A recommendation policy $\pi$ is $\delta$-BIC if and only if its LP-representation lies in $\bic_\delta[S]$.
\end{claim}

We carry over the notion of max-support policies: a \emph{$\delta$-max-support policy} $\pimax$ is a $\delta$-BIC recommendation policy $\pimax$ such that $\SEx_s[S] = \support(\pimax(s))$ for each feasible signal $s\in\cX$. Assuming there exists a $\delta$-BIC policy, a $\delta$-max-support policy can be computed by a version of
Algorithm~\ref{alg:MaxSupport}, where in the linear program~\eqref{eq:LP-Xas} the feasible set is $\bic_\delta[S]$ rather than $\bic[S]$. Let us call this modified algorithm $\maxSupportN$. Thus:

\begin{claim}\label{cl:max-support-policyN}
If there exists a $\delta$-BIC policy, then there exists a $\delta$-max-support policy $\pimax$ with
    \[ \pmin(\pimax)\geq \psig^\delta[S]\,/\,(|\cA|\cdot |\cX|).\]
It can be computed by $\maxSupportN$ in time polynomial in
$|\cA|$, $|\cX|$, and $|\Theta|$.
\end{claim}

Here $\psig^\delta[S]$ is defined as the natural extension of $\pmin[S]$ from  \refeq{eq:det-pmin}:
\begin{align}\label{eq:noisy-pmin}
\psig^\delta[S]
    := \quad\min_{s\in \cX,\; a\in \SEx_s[S]}\quad
            \max_{x\in\bic_\delta[S]}\quad x_{a,s}.
\end{align}


\subsection{Approximating the expected utilities}
\label{sec:approx-utils}

Let us discuss how to to approximate the expected utilities using multiple samples of stochastic utilities, and how to use the resulting approximate signal. In particular, we introduce an important tool (Lemma~\ref{lem:betabic-denoise}): a fact about approximate signals which we use throughout this section.

\xhdr{Notation.} Formulating this discussion in precise terms requires some notation. Fix subset $B\subset\cA$ of joint actions. Given state $\theta$, the expected utilities of joint actions in $B$ form a table
\begin{align*}
U(B,\theta) := \rbr{
        \rbr{ f(a,\theta); u_1(a,\theta) \LDOTS u_n(a,\theta)}:\; a\in B }.
\end{align*}
Recall that in Section~\ref{sec:rec-game-props}, full information pertaining to the joint actions in $B$ is expressed as a signal
    $\signal(B) = (B,\, U(B,\trueState))$,
as defined in \eqref{eq:signalB-def}.

Now, let $\theta=\trueState$ be the realized state. Let an \emph{IID utility vector} for a joint action $a\in B$ be an independent sample from distribution $\cD_{(a,\theta)}$. Recall from Section~\ref{sec:model} that such sample is a random vector in $[0,1]^{n+1}$ whose expectation is $U(a,\theta)$. Further, the \emph{$d$-sample} for $B$, $d\in\NN$, is a collection
    $U_B = \rbr{ v(a,j):\; a\in B, j\in[d] }$,
where each $v(a,j)$ is an IID utility vector for the corresponding joint action $a$. We would like to use the sample $U_B$ to approximate $U(B,\theta)$.

\xhdr{Approximation procedure.} Our approximation procedure is fairly natural. We input a $d$-sample $U_B$ for some subset $B\subset \cA$. Then we compute average utilities across the $d$ samples:
\begin{align*}
\overline{U} = \textstyle \rbr{ \tfrac{1}{d}\sum_{j=1}^d v(a,j):\; a\in B }.
\end{align*}
Then we map these averages to the state $\theta$ which provides the best fit for the averages:
\begin{align*}
\widehat{\theta} = \argmin_{\theta\in\Theta} \left\| \overline{U}- U(B,\theta)   \right\|_{\infty},
\end{align*}
where the ties are broken uniformly at random (any fixed tie-breaking rule would suffice).

We will represent the output of this procedure as a signal
    $\denoise(U_B) = \rbr{ B,\; U(B,\widehat{\theta}) }$.
Note that such signal has the same structure as the ``correct" signal $\signal(B)$, in the sense that $\denoise(U_B) \in \support(\signal(B))$, so we can directly compare the two signals.

\xhdr{Properties of the approximation.} First, we observe that the approximate signal $\denoise(U_B)$ is exactly equal to the correct signal $\signal(B)$ with high probability, as long as the number of samples is large compared to $\log(n\,|B|)$ and the inverse of the separation parameter $\myGap$.%

\begin{lemma}\label{lem:denoiseacc}
Fix subset $B\subseteq \cA$ of joint actions.  Let $U_B$ be a $d$-sample for $B$ such that $d$ is large enough, namely
    $d \geq \myGap^{-2}\cdot \ln\left( 2n\,|B|/\beta\right) $
for some $\beta>0$, where $\myGap$ is the separation parameter. Then
\begin{align}\label{eq:lem:denoiseacc}
\Pr\sbr{ \signal(B) = \denoise(U_B) \mid \trueState = \theta } \geq 1 - \beta,
\quad \forall \theta\in\Theta.
\end{align}
\end{lemma}

\noindent (This is an easy consequence of \emph{Chernoff-Hoeffding Bound}, see Appendix~\ref{app:pf:lem:denoiseacc}).

The crucial point here is that $\denoise(U_B)$ can be used instead of $\signal(B)$ in the recommendation game. More specifically, if $\pi$ be a $\delta$-BIC recommendation policy for signal $\signal(B)$, then it is also a BIC recommendation policy for signal $\denoise(U_B)$, as long as the approximation parameter $\beta$ is sufficiently small compared to the minimal probability $\pmin(\pi)$.

\begin{lemma}\label{lem:betabic-denoise}
\asedit{Fix subset $B\subseteq \cA$ of joint actions.  Let $U_B$ be a $d$-sample for $B$ which satisfies \eqref{eq:lem:denoiseacc} for some $\beta>0$.}
Let $\pi$ be a $\delta$-BIC recommendation policy for signal $\signal(B)$. Then $\pi$ is also a BIC recommendation policy for signal $\denoise(U_B)$, as long as
    $ \beta \leq \tfrac{\delta}{2} \cdot \pmin(\pi) / |\Theta|$.
\end{lemma}

We derive Lemma~\ref{lem:betabic-denoise} as a corollary of Lemma~\ref{lem:betabic-denoise-signal}, a general fact about approximate signals in a recommendation game, which we develop in the next subsection.

\subsection{Approximate signals in a recommendation game}
\label{sec:signal-approx}

We abstract \eqref{eq:lem:denoiseacc} as a property of two coupled signals, and generalize Lemma~\ref{lem:betabic-denoise}.

\begin{definition}\label{def:signal-approx}
Consider a recommendation game with coupled signals $S, S'$, where the two signals have
same support $\cX$. Signal $S'$ is called a \emph{$\beta$-approximation} for $S$, $\beta\in(0,1)$, if
\begin{align*}
\Pr\sbr{ S = S' \mid \trueState=\theta } \geq 1 - \beta,
\quad \forall \theta\in\Theta,
\end{align*}
where the randomness is taken over the realization of $(S,S',\trueState)$.
\end{definition}

\begin{lemma}\label{lem:betabic-denoise-signal}
Consider the setting of Definition~\ref{def:signal-approx}. Let $\pi$ be a $\delta$-BIC recommendation policy for signal $S$. Then $\pi$ is also a  BIC recommendation policy for signal $S'$, as long as
    $\beta \leq \tfrac{\delta}{2} \cdot \pmin(\pi) / |\cX|$.
\end{lemma}

As an intermediate step, we will show the following technical lemma.

\begin{lemma}\label{lem:bounddiff}
Consider the setting of Definition~\ref{def:signal-approx}. Let $g$ be a random
variable in the same probability space as $(S,S',\trueState)$, with bounded range $[0, H]$. Then
\begin{align}\label{eq:lem:bounddiff}
\left|\;
    \Pr[S = s]\cdot \E\sbr{g \mid S= s} \;-\;
    \Pr[S' =s]\cdot \E\sbr{g \mid S'= s} \;\right|
\leq \beta \,H,
\quad\forall s\in\cX.
\end{align}
\end{lemma}

\begin{proof}
Let $\cE$ denote the event of $\left( S = S'\right)$, which occurs
with probability at least $1-\beta$ by definition. We will also write
$\neg\cE$ to denote the event of $\left(S \neq S' \right)$, and we
know $\Pr[\neg \cE] \leq \beta$. Observe that the event $\left( S' = s
\right) \wedge \cE$ is equivalent to the event $\left( S = s\right)
\wedge \cE$. Also note that $g\in [0, H]$, so for each $s\in \cX$, we
could write
\begin{align*}
\Pr[S'=s] \cdot \E[g \mid S'=s]
=& \Pr[\cE,\, S'=s] \cdot \E[g \mid \cE,\, S'=s]
    + \Pr[\neg\cE,\,S'=s] \cdot \E[g \mid \neg\cE,\,S'=s] \\
\leq& \Pr[\cE,\, S'=s] \cdot \E[g \mid \cE,\, S'=s]
    + \Pr[\neg\cE] \cdot \E[g \mid \neg\cE,\,S'=s] \\
\leq& \Pr[\cE,\, S=s] \cdot \E[g \mid \cE,\, S = s] + \beta\, H\\
\leq& \Pr[\cE,\,S=s] \cdot \E[g \mid \cE,\, S=s]
    + \Pr[\neg\cE,\, S=s] \cdot \E[g \mid \neg\cE,\,S=s] + \beta \, H\\
=& \Pr[S=s] \cdot \E[g \mid S = s] + \beta\, H
\end{align*}
Similarly,
\begin{align*}
\Pr[S'=s] \cdot \E[g \mid S'=s]
=& \Pr[\cE,\, S'=s] \cdot \E[g \mid \cE,\, S'=s]
    + \Pr[\neg\cE,\,S'=s] \cdot \E[g \mid \neg\cE,\,S'=s] \\
\geq& \Pr[\cE,\, S'=s] \cdot \E[g \mid \cE,\, S'=s]\\
\geq& \Pr[\cE,\,S=s] \cdot \E[g \mid \cE,\,S=s]\\
\geq & \Pr[\cE,\,S=s] \cdot \E[g \mid \cE,\,S=s]
    + \Pr[\neg\cE,\,S=s] \cdot \E[g \mid \neg\cE,\,S=s] -\beta \, H\\
\geq& \Pr[S = s] \cdot \E[g \mid S = s] - \beta\, H. \qedhere
\end{align*}
\end{proof}

\begin{proof}[Proof of Lemma~\ref{lem:betabic-denoise-signal}]
For any agent $i\in [n]$ and any joint action $a\in \cA$, the utility
$u_i(a, \theta)$ is a random variable with bounded range $[0, 1]$. In particular, we obtain \eqref{eq:lem:bounddiff} for random variable $g=u_i(a,\trueState)$.

Let $x$ be the LP-representation of policy $\pi$ from the lemma statement. Pick some action $a_i\in \cA_i$ such that $\Pr[\pi_i(S) = a_i] > 0$, and some other action
$a_i'\in \cA_i\setminus \{a_i\}$. Denote
\begin{align*}
    W(a_{-i};\,\trueState) = u_i(a_i, a_{-i};\, \trueState) - u_i(a'_i,a_{-i}, ;\,\trueState),
\end{align*}
where $a_{-i}\in \cA_{-i}$ is a joint action of all agents but $i$. Then

\begin{align*}
\sum_{a_{-i}\in \cA_{-i},\; s\in \cX}
\Pr[S = s] \cdot
\E\left[ W(a_{-i};\,\trueState) \mid S=s \right]\, x_{a, s}
    \geq \delta \cdot \sum_{a_{-i} \in \cA_{-i},\; s\in \cX} \Pr[S = s]\; x_{a,s}
\geq \delta \cdot \pmin(\pi)
\end{align*}

To show that policy $\pi$ is a BIC policy for signal $S'$, it suffices to show that
    $\E\sbr{ W(\pi_{-i}(S) \mid \pi_i(S) = a_i}$
is non-negative, and we could write
\begin{align*}
&\sum_{a_{-i}\in \cA_{-i},\; s\in \cX} \Pr[S' = s]
    \cdot \E\sbr{ W(a_{-i};\,\trueState) \mid S' = s }\, x_{a, s}\\
\geq &\sum_{a_{-i}\in \cA_{-i},\; s\in \cX}
    \rbr{ \Pr[S = s]
    \cdot \E\sbr{ W(a_{-i};\,\trueState) \mid S = s }
    - 2\beta}\, x_{a, s} \qquad
    &\text{(\eqref{eq:lem:bounddiff} with $g=u_i(a,\trueState)$)}\\
= & \sum_{a_{-i}\in \cA_{-i},\; s\in \cX} \Pr[S = s]
    \cdot \E\sbr{ W(a_{-i};\,\trueState) \mid S = s }\, x_{a, s} - 2\beta\, |\cX|\\
\geq & \delta\,\pmin(\pi) - 2 \beta\,|\cX| \geq 0 \qquad &\mbox{(by assumptions in Lemma~\ref{lem:betabic-denoise-signal})}
\end{align*}
Therefore, we know that $\pi$ is BIC w.r.t. the signal $S$.
\end{proof}

\subsection{Maximal exploration in a recommendation game}
\label{sec:delta-max-rec}

Let $S$ be a signal with support $\cX$. We present a BIC subroutine $\SMaxEx$ that given an approximation signal $\hat S$ to the signal $S$, explores all the
joint actions in the set $\cA' = \SEx[S]$. The
subroutine collects multiple utility samples of each explored
joint action, which allows us to construct a new signal
approximation for the signal $\signalD(\cA')$. 
The pseudocode can be found in Algorithm~\ref{alg:SMaxEx}.

%

\begin{algorithm}[hbt]
\caption{$\SMaxEx(\cU, \sigStruc, \hat S, \beta, \beta')$:
  maximal exploration given an approximate signal $\hat S$.}

\begin{algorithmic}
\STATE {\bf Input:} utility structure $\cU$, signal structure
$\sigStruc$ with associated signal $S$;
\STATE \myTab $\hat S$ is a $\beta$-signal approximation to $S$, where $\beta\in(0,1)$ is the input approximation parameter;
\STATE \myTab $\beta'\in(0,1)$ is the output approximation parameter
\vspace{2mm}

\COMMENT{compute the parameters}
\STATE let $s\in\cX$ be the realization of signal $\hat S$
\STATE $x \leftarrow \maxSupportN(\cU,\sigStruc)$
    \COMMENT{LP-representation of $\pimax$}
\STATE $B \leftarrow \{ a:\in \cA:\; x_{a,s}>0\} $
    \COMMENT{$\delta$-signal-explorable set $\SEx_s[S]$}
\STATE $\pmin \leftarrow \min(x_{a,s'}:\; a\in B,\; s'\in\cX)$
    \COMMENT{computes $\pmin(\pimax)$}
\STATE $T\leftarrow \max(1+|\cA|,\; \lceil 1/\pmin\rceil)$ and $R \gets \frac{1}{\myGap^2}\ln\left(\frac{2n |B|}{\beta'} \right)$
    \COMMENT{\#rounds and \#meta-rounds}
\STATE Pick injective function $\tau:B\to [T]$ u.a.r. from all such functions.
    \COMMENT{``dedicated rounds"}
\STATE
    $\cD(a) \leftarrow \frac{x_{a,s}-1/T}{T-|B|}\quad \forall a\in B$.
    \COMMENT{distribution \eqref{eq:remainder-distribution} for non-dedicated rounds}
\vspace{2mm}
\STATE {Initiate a set $U_B$ for storing stochastic utilities samples.}
\STATE
\COMMENT{issue the recommendations}
\FOR{meta-rounds $r = 1 \ldots R$}
\FOR{rounds $t=1 \ldots T$}
\STATE {\bf if}
        $t=\tau(a)$ for some $a\in B$
    {\bf then}
        $a_t\leftarrow a$
    {\bf else}
         choose $a_t$ from distribution $\cD$
\STATE Recommend action $a_t$, store the corresponding utilities in $U_B$
\ENDFOR
\ENDFOR
\vspace{2mm}
\STATE {\bf Output:} \asedit{signal $S' = (B, \widehat U)\gets \denoise(B, U_B)$}
\end{algorithmic}

\label{alg:SMaxEx}
\end{algorithm}

\begin{lemma}\label{cor:smax-bic}
\asedit{Let
    $\Lambda = \psig^\delta[S]\,/\,(|\cA|\cdot |\cX|)$
be the lower bound on $\pmin(\pimax)$ from Claim~\ref{cl:max-support-policyN}. }
The subroutine~$\SMaxEx$ is BIC if the input
approximation parameter satisfies 
    \asedit{$\beta \leq \tfrac{\delta}{2}\cdot \Lambda/|X|$.
The total number of rounds is at most
     $\frac{\ln\rbr{2n\,|\cA|/\beta'}}{\Lambda \cdot\myGap^2}$.}
\end{lemma}

\begin{proof}
Let $\pi$ be the policy computed by $\maxSupportN(\cU,\sigStruc)$ within our
instantiation of $\SMaxEx$.  Since $\maxSupportN$ is a
composition of $R$ copies of $\pi$, the BIC property follows 
from Claim~\ref{cl:warmup-composition} and Lemma~\ref{lem:betabic-denoise-signal}. The number of rounds is immediate from the specification of the subroutine, plugging in $\Lambda\leq \pmin$.
\end{proof}

Note that if our algorithm has access to deterministic utilities, in
the end it will be able to construct a signal $\signalD(B)$, where $B$
is the (random) set of actions explored by the algorithm. We now show
that even though our algorithm only has access to stochastic
utilities, the output signal will be a $\beta$-signal approximation to
$\signalD(B)$.

\begin{lemma}\label{lem:per-round-beta}
The output signal $S'$ in $\SMaxEx$ is a $(\beta + \beta')$-approximation to
    $\signalD\rbr{\SEx[S]}$,
where $\beta$ and $\beta'$ are, resp., the input and output approximation parameters.
\end{lemma}

\begin{proof}
Fix any feasible state $\theta$. First, 
    $\Pr[S = \hat S]\geq 1 - \beta$,
where the probability is over the randomness in $S$ and $\hat S$. 
Conditional on $S = \hat S$, the subroutine explores the subset of joint actions 
    $B =\SEx_S[S]$
that we are interested in. Then, the lemma follows from the accuracy
guarantee of $\denoise$ (Lemma~\ref{lem:denoiseacc}).
\end{proof}

\subsection{Maximal exploration in a repeated game}
\label{sec:delta-max}

Let us presents the maximal-exploration subroutine analogous to the one in ~\Cref{sec:det-maxExplore}. It (only) accesses stochastic utilities, and explores all joint actions that are explorable by any $\delta$-strictly BIC policy that has access to
deterministic utitlities. 



Similar to $\IndMax$, the subroutine proceeds in phases $\ell = 1,
2, 3, \ldots , |\cA|$. In each phase $\ell$, we use the
approximate signal $S^\beta_\ell$ to perform maximal exploration by
instantiating the function $\SMaxEx$. This allows us to
construct an approximate signal for the next phase. 
The pseudocode can be found in Algorithm~\ref{alg:sinduct}.


\begin{algorithm}
  \begin{algorithmic}
    \STATE{{\bf Input}: utility structure $\cU$, confidence parameter $\beta\in(0,1)$}
    \STATE{{\bf Initialize}:
$\beta' = \beta / |\cA|$~~~and~~~
$\hat S_1 = S_1 = \sigStruc_1 = \bot$. \COMMENT{``phase-$1$ signal" is empty}}

\vspace{2mm}
\STATE \COMMENT{compute signal structure $\sigStruc_{\ell}$ for each signal $S_{\ell}$}
    \ascomment{new loop}
\FOR{each phase $\ell=1,2 \LDOTS |\cA|$}
    \STATE $x \leftarrow \maxSupportN\rbr{\cU,\sigStruc_\ell}$
        \COMMENT{LP-representation of a max-support policy}
    \asedit{ 
    \FOR{each state $\theta\in\Theta$}
        \STATE  $s\leftarrow S_\ell(\theta)$
            \COMMENT{value of signal $S_\ell$ if $\trueState=\theta$}.
        \STATE $S_{\ell+1}(\theta)\leftarrow \signalD\rbr{\{ a\in\cA: x_{a,s}>0\}}$
            \COMMENT{if $\trueState=\theta$ then
            $\SEx[S_\ell] = \cbr{ a\in\cA: x_{a,s}>0 }$.}
    \ENDFOR
    } 
    \STATE $\sigStruc_{\ell+1} \leftarrow \rbr{S_{\ell+1}(\theta):\; \theta\in\Theta}$.
\ENDFOR

\vspace{2mm}
\STATE \COMMENT{play the repeated game}
\FOR{each phase $\ell=1,2 \LDOTS |\cA|$}
    \STATE \COMMENT{compute phase-$(\ell+1)$ approximate signal to
        $S_{\ell+1} = \signalD\sbr{\SEx[S_\ell]}$}
     \STATE $\hat S_{\ell+1} = (B_{\ell+1}, \widehat U_{\ell+1}) \gets 
        \SMaxEx\rbr{ \cU,\; \sigStruc_{\ell},\; \hat S_\ell,\; \asedit{\beta/2,\; \beta/2}}$.
\ENDFOR

\vspace{2mm}
\STATE{{\bf Output}: final signal $\hat S_{|\cA|+1}$}
    \end{algorithmic}

\caption{Subroutine $\SIndMax(\cU , \beta)$: $\delta$-maximal exploration.}
\label{alg:sinduct}
\end{algorithm}

\asedit{To state the immediate guarantees about this subroutine, let us denote
\begin{align}\label{eq:noisy-unifLB}
\unifLB = \frac{\min_{\ell\leq |\cA|}\;\psig^\delta[S_\ell]}{|\cA|\cdot |\cX|}.
\end{align}
By Claim~\ref{cl:max-support-policyN}, it is a lower bound on 
    $\pmin(\pimax)$
for each signal $S_\ell$.}

\begin{lemma}\label{lem:sind-bic}
$\SIndMax(\cU, \beta)$ is BIC as long as the parameter satisfies
    \asedit{$\beta \leq \delta\cdot \unifLB/|\Theta|$}.
\asedit{The total number of rounds is at most 
     $T_{\beta,\delta} = \frac{|\cA|\cdot \ln\rbr{4n\,|\cA|/\beta}}{\unifLB \cdot\myGap^2}$.}
The per-round running time is
  $\poly\rbr{|\cA|\cdot |\Theta|}$.
\end{lemma}

\begin{proof}
$\SIndMax$ could be viewed as a composition of subroutines
$\SMaxEx$ such that for each $\ell\in [|\cA|]$ the instantiation in
phase $\ell+1$ is a valid sequel of the one in phase $\ell$.  Also
observe that the support of each signal $S_\ell$ has size at most 
$|\Theta|$. Then Lemma~\ref{cor:smax-bic} implies that 
instantiation of $\SMaxEx$ in each phase is BIC as long as 
    $\beta\leq \tfrac{\delta}{2}\cdot \unifLB/|\Theta|$.
It follows that $\SIndMax(\cU, \beta)$ is also BIC as a composition of BIC subroutines (by Claim~\ref{cl:warmup-composition}).
\end{proof}

Finally, we will show that $\SIndMax$ is $\delta$-maximally exploring
and outputs a $\beta$ approximate signal for
$\signalD(\cA_\theta^{\delta})$.

\begin{lemma}\label{lm:noisy-approx}
The subroutine $\SIndMax$ is $\delta$-maximally-exploring with
probability at least $1 - \beta$ over the randomness of the stochastic
utilities, and outputs a signal $\hat S_{|\cA|+1}$ that is a
$\beta$-signal approximation to $\signalD(\cA_{\trueState}^{\delta})$.
\end{lemma}

\begin{proof}
Fix any game state $\theta$ as our true state $\trueState$.
Let $\pi$ be any $\delta$-BIC iterative recommendation policy with
access to~\emph{deterministic utilities}. Our goal is to show that all
actions explored by $\pi$ are also explored by $\SIndMax$ with
probability at least $1 - \beta$.

Note that in each phase $\ell$ of the $\SIndMax$, we know
by Lemma~\ref{lem:per-round-beta} that except with probability at most
$\beta'$ over the randomness of the stochastic utilities, the output
signal satisfies
$$\hat S_{\ell + 1} = \signalD(\SEx[S_{\ell}])$$
For the remainder of the proof, we will condition on this event over
all phases of $\SIndMax$, which by union bound occurs with probability
at least $1 - |\cA|\beta' = 1 -\beta$.

Denote the internal random seed in $\pi$ by $\omega$. We will think of
$\omega$ being drawn ahead of first round of our Bayesian exploration
game and fixed throughout the game. Let $A_t$ be the set of actions
explored by $\pi$ in the first $(t-1)$ rounds, which is determined by
the realization of $\omega$ and the true state $\trueState$ (since
$\pi$ sees deterministic utitlities). We will represent $\pi$ as a
sequence of $\left(\pi^{(t)} : t\in [T] \right)$, where each
$\pi^{(t)}$ is a BIC recommendation policy which inputs the signal
$S_t^*$ which consists of the random seed and the observations so far:
$S^*_t = (\omega, \signalD(A_t))$, and is determined given that
signal. In each round $t$, we will denote the recommended joint action
by $\pi^{(t)}(S_t^*)$. We will prove the following claim using
induction on round $t$:
\begin{equation}\label{eq:puppy}
\mbox{for each round } t, \mbox{ there exists phase } \ell\in \NN \mbox{
  such that } \pi^{(t)}(S_t^*)\in B_\ell,
\end{equation}
where $B_\ell$ is the random set of joint actions explored by subroutine $\SIndMax$ in phase $\ell$.

 For the base case ($t=1$), we know that $A_1 = \emptyset$, so the
 signal $S_1^* = (\omega, \perp)$. Since the random seed is
 independent of the game state, we know that the empty signal
 ($\perp$) is at least as informative as the signal $S_1^*$. This
 means
\begin{align*}
  \pi^{(1)}(S_1^*) &\in \SEx[S_1^*] \qquad   (\mbox{by Definition~\ref{def:sex}})&\\
  &\subseteq \SEx[\perp] \qquad (\mbox{by Lemma~\ref{lm:coupled-explore}})&\\
  &=B_1 \qquad (\mbox{by definition of } B_1)&
\end{align*}

For the induction step, assume that~(\ref{eq:puppy}) holds for all
rounds $t < t_0$, for some $t_0\in \NN$. This implies
$A_{t_0}\subseteq B_\ell$ for some phase $\ell$ of $\SIndMax$, and so
the signal $\hat S_\ell$ is at least as informative as the signal
$S_t^*=(\omega, \signalD(A_{t_0}))$. It follows that
\begin{align*}
  \pi^{(t_0)}(S_{t_0}^*) &\in \SEx[S_{t_0}^*] \qquad (\mbox{by Definition~\ref{def:sex}})&\\
  &\subseteq \SEx[\hat S_\ell]\qquad(\mbox{by Lemma~\ref{lm:coupled-explore}})&\\
  &=B_\ell\qquad(\mbox{by definition of } B_\ell)&
\end{align*}
This gives a proof for the claim in~\eqref{eq:puppy}. By the same
reasoning of Claim~\ref{cl:det-Bl}, the $(B_\ell:\;\ell\in\NN)$ are
non-decreasing in $\ell$, and identical for all $\ell\geq |\cA|+1$, so
we have shown that all joint actions explored by $\pi$ must be
contained in $B_{|\cA|+1}$. It follows that $\SIndMax$ is
$\delta$-maximally exploring given access to stochastic utitlities,
and $\hat S_{\ell+1} = \signalD(\cA_{\trueState}^\delta)$.
\end{proof}

\subsection{Putting this all together: proof of  Theorem~\ref{thm:model-noisy}}
\label{sec:stoc-wrapup}

\asedit{We show that $\SIndMax$ leads to a near-optimal iterative recommendation policy. The remaining piece is that whenever a subroutine outputs an approximate signal $\hat S$ for $\signalD(\cA_{\trueState}^\delta)$, we can use $\hat S$ afterwards to
obtain near-optimal rewards.}

\begin{lemma}\label{lem:signal-opt}
Let $\delta, \beta \in (0, 1/2)$ and $\sigma$ be \asedit{any subroutine}
\asmargincomment{$\sigma$ does not need to be max-exploring here!}
with duration $T_\sigma$ that outputs a
$\beta$-signal approximation $\hat S$ to the signal $S^* =
\signalD(\cA_\theta^{\delta})$. Let $\pi^*$ be an optimal $\delta$-BIC
recommendation policy for signal $S^*$, and $\pi$ be the composition
of subroutine $\sigma$ followed by $T - T_\sigma$ copies of
$\pi^*(\hat S)$. Then $\pi$ is BIC as long as 
\asedit{$\beta \leq \tfrac{\delta}{2}\cdot \unifLB/|\Theta|$,
where $\unifLB$ is from \eqref{eq:noisy-unifLB}}. Moreover,
 \[
 \REW(\pi) \geq (T - T_\sigma) (1 - \beta)\, \optdet(\PiBICdet[\delta]).
 \]
\end{lemma}

\begin{proof}
Since the support of $S^*$ is at most the number of states $|\Theta|$, it
follows from Lemma~\ref{lem:betabic-denoise} that $\pi^*\in \bic[\hat S]$,
and since $\sigma$ is BIC, we also have that $\pi$ is BIC
by Claim~\ref{cl:warmup-composition}.

In the following, we will write $s^*$ and $\hat s$ to denote the
realization of signals $S^*$ and $\hat S$ respectively, and use
$\psi_*$ to denote the joint distribution over the two signals $S^*$,
$\hat S$ and the state $\theta$. By the definition of $\beta$-signal
approximation in Definition~\ref{def:signal-approx}, we have for each $\theta\in \Theta$,
\[
\psi_*\rbr{ \hat S = S^* \mid \trueState = \theta } \geq 1 - \beta.
\]
For $t > T_\sigma$, we can write
the reward at round $t$ as
\begin{align*}
\E_{\pi, \psi_*} \sbr{ f(\pi_t(\hat s), \trueState)}
&= \sum_{\theta\in \Theta} \psi(\theta)
    \E_{\pi, \psi_*}\sbr{ f(\pi_t(\hat s), \theta) \mid \trueState = \theta }\\
&=\sum_{\theta\in \Theta} \psi(\theta)
    \biggl( \psi_*(\hat S = S^* \mid \trueState = \theta)\cdot
    \E_{\pi, \psi_*}\sbr{ f(\pi_t(\hat s), \theta) \mid \trueState = \theta, \hat S =
  S^* } \\
  &\qquad\qquad+ \psi_*(\hat S \neq S^*\mid\theta)\cdot \E_{\pi, \psi_*}\sbr{
  f(\pi_t(\hat s), \theta) \mid \trueState = \theta, \hat S \neq S^*
  }\biggr)\\
&\geq\sum_{\theta\in \Theta} \psi(\theta)(1-
\beta)\E_{\pi, \psi_*}\sbr{ f(\pi_t(\hat s), \theta) \mid \trueState = \theta, \hat S =
  S^* } \\
&= (1-\beta)\sum_{\theta\in \Theta} \psi(\theta) \E_{\pi, \psi_*}\sbr{ f(\pi_t(s^*), \theta) \mid \trueState = \theta, \hat S =
  S^* }
\end{align*}

Note that the random variable $\E_{\pi, \psi_*}\sbr{ f(\pi_t(s^*),
  \theta) \mid \trueState = \theta }$ is independent of event 
$\cbr{\hat S= S^*}$ since the reward is fully determined by the
state and the randomness of $\pi$. In particular, for each
$\theta\in \Theta$,
$$\E_{\pi, \psi_*}\sbr{  f(\pi_t(s^*), \theta) \mid \trueState = \theta, \hat S = S^* } = \E_\pi\sbr{ f(\pi_t(S^*(\theta), \theta)) }$$
In other words,
\[
\E_{\pi, \psi_*} \sbr{ f(\pi_t(\hat s), \theta)}
\geq (1-\beta)\sum_{\theta\in \Theta} \psi(\theta)
    \E_\pi\sbr{ f(\pi_t(S^*(\theta), \theta)) }
= (1 - \beta)\;\REW^*[S^*]
\]
Note that we also have $\REW^*[S^*] \geq
\optdet(\PiBICdet[\delta])$
by Claim~\ref{cl:restricted-subset}. Therefore, in the last $(T -
T_\sigma)$ rounds, the expected reward of $\pi$ is at least 
    $(1 - \beta)\,\optdet(\PiBICdet[\delta])$.
\end{proof}

\asmargincomment{simplified the concluding argument, spelled out the "constants"}
\asedit{To derive Theorem~\ref{thm:model-noisy}, we invoke Lemma~\ref{lem:signal-opt} for subroutine
    $\SIndMax(\cU, \beta)$.
Its duration is upper-bounded by $T_{\beta,\delta}$ in Lemma~\ref{lem:sind-bic}, and it is guaranteed to output a $\beta$-signal-approximation to
    $\signalD(\cA_{\trueState}^{\delta})$ 
by Lemma~\ref{lm:noisy-approx}. Taking
    $\beta = \min\rbr{ \tfrac{1}{T},\; \tfrac{\delta}{2}\cdot \unifLB/|\Theta|}$,
we obtain
 \begin{align}\label{eq:noisy-final}
 T\, \optdet\rbr{\PiBICdet[\delta]} - \REW(\pi) 
    \leq  1+ \frac{|\cA|\cdot \ln\rbr{4n\,|\cA|/\beta}}{\unifLB \cdot\myGap^2},
 \quad\text{where $\unifLB$ is from \eqref{eq:noisy-unifLB}}.
 \end{align}
By Lemma~\ref{lm:model-benchmarks} the benchmark $\optdet\rbr{\PiBICdet[\delta]}$ is at least as strong as the one used in Theorem~\ref{thm:model-noisy}.
}

\section{Conclusions and open questions}
\label{sec:conclusions}

We introduce Bayesian Exloration, a fundamental model which captures incentivizing exploration in Bayesian games, and resolve the first-order issues in this model: explorability and constant/logarithmic regret. Our policies are computationally efficient, in that the per-round running time is polynomial in the input size under a generic game representation. The number of rounds is polynomial in the input size (in the same sense), and also depends on the Bayesian prior.

We make some simplifying assumptions for the sake of tractability.
Let us discuss them in more detail, and outline the corresponding directions for future work.

First, we assume that the game does not change over time, even though there is a fresh set of agents in each round. A more realistic model would have each round and/or each agent identified by "context", such as the current weather and the source-destination pair in the traffic routing example, and allow the tuple of contexts to change from one round to another. Then we (still) have a parameterized game matrix that the principal wishes to learn, but the parametrization is now by states \emph{and} contexts. Absent BIC restriction, this is special case of ``contextual bandits", a generalization of multi-armed bandits that is well-studied in machine learning
\citep[\eg see Chapter 8 in][]{slivkins-MABbook}.
Extending Bayesian Exploration to a model with contexts would be parallel to extending from multi-armed bandits to contextual bandits in the machine learning literature. For the special case of a single agent per round, such extensions have been accomplished in prior work \citep{ICexploration-ec15-conf,ICexploration-ec15} and subsequent work \citep{Jieming-multitypes18}.


\asedit{Second, we assume that the principal can commit to a particular policy, and the agents rationally respond to the incentives created by this policy. Such assumptions are standard in theoretical economics, and tremendously useful in many scenarios. Mitigating them tends to be very challenging. A subsequent paper \citep{Jieming-unbiased18} accomplishes this for the basic version of Bayesian Exploration with one agent per round. A part of the challenge is that it is fundamentally unclear what are the ``right" assumptions on economic behavior that one should strive for.}

\asedit{Third, our results for stochastic utilities only compete against $\delta$-BIC policies, for a given $\delta>0$, whereas the standard notion of BIC corresponds to $\delta=0$. Essentially, this is because our approach requires at least $\log(1/\delta)$ samples from each explorable joint action. While such approach cannot handle $\delta=0$, one could hope to compete against $\delta$-BIC policies for all $\delta>0$ at once, call such policies \emph{strictly-BIC}. This is equivalent to competing against $\delta$-BIC policies for a sufficiently small $\delta = \delta_{\min}>0$, which needs to be found by the algorithm. However, such $\delta_{\min}$ may depend on the true state $\trueState$ (because the set of explorable joint actions may depend on $\trueState$, as per Remark~\ref{rem:explorable}), so it cannot be computed in advance. Estimating $\delta_{\min}$ on the fly may be feasible, but it appears to require the $\SMaxEx$ subroutine to have a data-dependent duration (because of the required number of samples mentioned above). This breaks our analysis on a very basic level, as per Remark~\ref{rem:composition}, because we cannot ensure that the composition of BIC subroutines is itself BIC. So, competing against strictly-BIC policies remains elusive.}

There is also a major assumption that we do \emph{not} make: we do not assume that the game matrix has a succinct representation, \eg via an (unknown) low-dimensional parameter vector. More generally, we do not make any ``inference assumptions" that would allow one to infer something about the rewards of one joint action from the rewards of other actions.%
\footnote{Similarly, "inference assumptions" do not appear in the fundamental results on multi-armed bandits \citep[such as][]{Gittins-index-79,bandits-ucb1,bandits-exp3}, and instead are relegated to the subsequent work.}
Absent such "inference assumptions", exploring all explorable joint actions is typically necessary to optimize the recommendation policy, and computation must take time polynomial in the size of the game matrix. In contrast, a succinct game representation may potentially allow for must faster exploration and/or computation.
\asedit{A subsequent paper \citep{IncentivizedRL} addresses this question in the special case of one agent per round, focusing on incentivizing exploration in reinforcement learning; here the succinct representation is given by the MDP.}  A simultaneous paper \citep{Dughmi-stoc16} addresses a similar question for computationally solving the Bayesian Persuasion game.

\OMIT{ 
Our results pave the way for future work in several directions. The most immediate direction is \emph{computational}: can we achieve polynomial per-round running time if the game has a succinct representation? A simultaneous paper \citep{Dughmi-stoc15} studies a similar question for Bayesian Persuasion. In terms of statistical guarantees, one may want to improve the regret bounds, \eg reduce the dependence of the asymptotic constants on the number of joint actions and parameters of the prior. In the \emph{economics} direction, it is appealing to address agent heterogeneity by incorporating idiosyncratic signals that can be observed and/or elicited by the principal.
} 

\clearpage
\addcontentsline{toc}{section}{References}

\bibliographystyle{plainnat}
\bibliography{refs,bib-abbrv,bib-AGT,bib-bandits,bib-slivkins,bib-competition}

\newpage
\appendix
\section{Monotonicity in information: formulations and proofs}
\label{app:coupled-signals}

Recall the recommendation game with coupled signals $S,S'$ such that $S$ is at least as informative as $S'$ (as per Definition~\ref{def:more-informative}). We prove the following lemma which unifies Lemma~\ref{lm:coupled-exploit} and Lemma~\ref{lm:coupled-explore}. We state this lemma for $\delta$-BIC policies, as it is also used for stochastic utilities in Section~\ref{sec:noisy}.

\begin{lemma}[monotonicity-in-information]
\label{lm:coupled}
Consider a recommendation game with coupled signals $S,S'$ such that
$S$ is at least as informative as $S'$. Then
\begin{itemize}
\item[(a)]
    $\REW^*[S]\geq \REW^*[S']$
\item[(b)] Fix $\delta\geq 0$. Then $\SEx[S']\subset \SEx[S]$ with probability $1$.
\end{itemize}
\end{lemma}


Throughout this section, signals $S,S'$ are as in Lemma~\ref{lm:coupled}, and $\cX,\cX'$ are their respective supports. All expectations are over the random choice of $(S,S',\trueState)$, and also the internal randomness of recommendation policies (if applicable).


We will use the notion of ``at least as informative" via the following corollary (which is immediate from Definition~\ref{def:more-informative}):

\begin{claim}\label{cl:redundancy}
For any function $h:\Theta\to\RR$ and any feasible signal $s\in S, s'\in S'$
\begin{align*}
\E\sbr{ h(\trueState) \mid S=s} = \E\sbr{ h(\trueState) \mid S=s,S'=s'}.
\end{align*}
\end{claim}

Given a policy $\pi'$ for signal $S'$, one can define the \emph{induced} policy $\pi$ for signal $S$ by setting
\begin{align*}
\Pr\sbr{\pi(s)=a} = \Pr\sbr{ \pi'(S')=a \mid S=s}
    \qquad \forall a\in\cA,\, s\in \cX.
\end{align*}

We use Claim~\ref{cl:redundancy} to derive a more elaborate technical property: essentially, that policies $\pi$ and $\pi'$ are equivalent when applied to any given function of
$h:\cA\times \Theta\to\RR$.

\begin{claim}\label{cl:redundancy-policy}
Let $\pi'$ be an arbitrary recommendation policy for signal $S'$, and let $\pi$ be the induced policy for signal $S$. Then for any function $h:\cA\times \Theta\to\RR$,
\begin{align}\label{eq:cl:redundancy-policy}
\E\sbr{ h(\pi(S),\trueState)} = \E\sbr{h(\pi'(S'),\trueState)}.
\end{align}
\end{claim}

\begin{proof}
Fix a feasible signal $s\in\cX$. Assume for now that policy $\pi'$ is deterministic. Then for any joint action $a\in\cX$ we have
\begin{align}
&\Pr[\pi(s) = a \mid S=s] \nonumber \\
&\qquad = \Pr[\pi(s) = a ]
    & \text{(because randomization in $\pi$ is independent)} \nonumber \\
&\qquad = \Pr[\pi'(S') = a \mid S=s]
    & \text{(by definition of induced policy)} \nonumber \\
&\qquad = \sum_{s'\in \cX'}
    \Pr[\pi'(s')=a] \cdot \Pr[S'=s' \mid S=s] \nonumber \\
&\qquad = \sum_{s'\in \cX':\; \pi'(s')=a} \; \Pr[S'=s' \mid S=s]
    & \text{(since $\pi'$ is deterministic)}. \label{eq:cl:reduncancy-pf}
\end{align}
Therefore,
\begin{align*}
&\E\sbr{h(\pi(S),\trueState) \mid S=s} \\
&\qquad = \sum_{a\in\cA}\; \E\sbr{h(a,\trueState) \mid S=s} \cdot \Pr[\pi(s)=a \mid S=s]  \\
&\qquad = \quad \sum_{a\in\cA,\; s'\in\cX':\; \pi'(s')=a}\;
    \E\sbr{h(a,\trueState) \mid S=s} \cdot  \Pr[S'=s' \mid S=s]
    &\text{(by \eqref{eq:cl:reduncancy-pf})}\\
&\qquad = \sum_{s'\in\cX'}\;
    \E\sbr{h(\pi'(s'),\trueState) \mid S=s} \cdot  \Pr[S'=s' \mid S=s] & \\
&\qquad = \sum_{s'\in\cX'}\;
    \E\sbr{h(\pi'(s'),\trueState) \mid S=s,\, S=s'} \cdot  \Pr[S'=s' \mid S=s]
    &\text{(by Claim~\ref{cl:redundancy})}\\
&\qquad = \E\sbr{h(\pi'(s'),\trueState) \mid S=s}.
\end{align*}
Taking expectations over the realizations of signal $S$, we obtain the claim (namely, \eqref{eq:cl:redundancy-policy}) for the special case when policy $\pi'$ is deterministic. For a randomized policy $\pi'$, we obtain the claim by taking expectation over all possible realizations of $\pi'$.
\end{proof}

We use Claim~\ref{cl:redundancy-policy} in two ways: to argue about the rewards and to argue about BIC.

\begin{corollary}\label{cor:redundancy-policy}
Let $\pi'$ be a policy for signal $S'$, and let $\pi$ be the induced policy for signal $S$. \begin{OneLiners}
\item[(a)] $\REW(\pi) \geq \REW(\pi')$.
\item[(b)] Fix $\delta\geq 0$. If $\pi'$ is $\delta$-BIC, then the induced policy $\pi$ is $\delta$-BIC, too.
\end{OneLiners}
\end{corollary}

\begin{proof}
For part (a), simply use Claim~\ref{cl:redundancy-policy} with function $h(a,\theta)=f(a,\theta)$.

For part (b), fix agent $i$ and any two actions $a_i, \tilde{a}_i\in \cA_i$. Let us use the following shorthand:
\begin{align*}
\Delta u(b,\theta)
    :=u_i\left( (a_i, b_{-i}) , \theta\right) -
        u_i\left( (\tilde{a}_i, b_{-i}), \theta \right),
        \quad b\in\cA,\, \theta\in\Theta.
\end{align*}
\noindent Use Claim~\ref{cl:redundancy-policy} with function
\begin{align*}
h(b,\theta) := \indicator{b_i=a_i} \cdot \Delta u(b,\theta),
\quad b\in\cA,\, \theta\in\Theta.
\end{align*}
It follows that
\begin{align*}
\E\sbr{ \indicator{\pi_i(S)=a_i} \cdot \Delta u(\pi(S); \trueState) }
=\E\sbr{ \indicator{\pi'_i(S')=a_i} \cdot \Delta u(\pi'(S'); \trueState) }.
\end{align*}
Consequently, since
    $\Pr[\pi_i(S)=a_i] = \Pr[\pi'_i(S')=a_i]$,
we have
\begin{align*}
\E\sbr{ \Delta u(\pi(S); \trueState)      \mid \pi_i(S)=a_i }
=\E\sbr{ \Delta u(\pi'(S'); \trueState) \mid \pi'_i(S')=a_i },
\end{align*}
whenever $\Pr[\pi_i(S)=a_i]>0$. The right-hand side of this equation is greater or equal to $\delta$ because policy $\pi'$ is $\delta$-BIC for signal $S'$. Since this holds for any agent $i$ and any two actions $a_i, \tilde{a}_i\in \cA_i$, the induced policy $\pi$ is $\delta$-BIC, too.
\end{proof}

\begin{proof}[Proof of Lemma~\ref{lm:coupled}]
To prove that
    $\REW^*[S]\geq \REW^*[S']$,
let $\pi'$ be an optimal policy for signal $S'$, and $\pi$ be the corresponding induced policy for signal $S$. Then by Corollary~\ref{cor:redundancy-policy} it follows that policy $\pi$ is BIC and has the same expected reward; therefore,
    $\REW^*[S]\geq \REW(\pi) = \REW(\pi') = \REW^*[S']$.

To prove that
     $\SEx[S']\subset \SEx[S]$,
let $\pi'$ be a $\delta$-max-support policy for signal $S'$, and $\pi$ be the corresponding induced policy for signal $S$. Policy $\pi$ is $\delta$-BIC by Corollary~\ref{cor:redundancy-policy}(b). Moreover, for all feasible signals $s\in\cX, s'\in\cX'$ such that
    $\Pr[S'=s'\mid S=s]>0$ we have:
\begin{align}
\Ex_s[S]
    &= \cbr{ a\in\cA:\, \Pr[\pi(s)=a]>0} \nonumber \\
    &= \cbr{ a\in\cA:\, \Pr[\pi'(S')=a \mid S=s]>0}
        &\text{(since $\pi$ is an ``induced policy")}\nonumber \\
    &\supset \cbr{ a\in\cA:\, \Pr[\pi'(s')=a]>0}
        &\text{(see below)} \label{eq:pf:lm:coupled}\\
    &= \Ex_{s'}[S']. \nonumber
\end{align}
It follows that $\SEx[S]\supset \SEx[S']$, as claimed.

It remains to prove that the inclusion $\supset$ holds in \eqref{eq:pf:lm:coupled}. Fix some joint action
 $a\in\cA$ such that $\Pr[\pi'(s')=a]>0$. Recall that
 $\Pr[S'=s'\mid S=s]>0$
by assumption. Then
\begin{align*}
\Pr[\pi'(S')=a \mid S=s]
    &\geq \Pr[S'=s'\mid S=s]\cdot \Pr[\pi'(S'=s')=a \mid S=s,\, S'=s']  \\
    &= \Pr[S'=s'\mid S=s]\cdot \Pr[\pi'(s')=a] \\
    &>0. \qedhere
\end{align*}
\end{proof}

\section{Approximating the expected utilities: proof of Lemma~\ref{lem:denoiseacc}}
\label{app:pf:lem:denoiseacc}

We will use the \emph{Chernoff-Hoeffding Bound}, a standard result on concentration of measure.

\begin{lemma}[Chernoff-Hoeffding Bound]\label{chbound}
Let $X_1,\ldots , X_n$ be i.i.d. random variables with $\Ex[X_i] =
\mu$ and $a\leq X_i\leq b$ for all $i$. Then for every $\delta > 0$,
\[
\Pr\left[\left|\frac{\sum_i X_i}{n} -\mu \right|\geq \delta\right]
\leq 2\exp\left(\frac{-2\delta^2 n}{(b-a)^2} \right).\]
\end{lemma}

To show that $S'$ is a $\beta$-approximation for $S$, we will
first show that for any fixed state $\theta$, with probability at
least $1 - \beta$ over the realization of the stochastic utilities, the
realizations of $S'$ matches with $S$. In particular, we will
show that the average utilities $\overline U = \left(\overline u_i(a)
\right)_{i\in [n], a\in \cA}$ are close to the expected utilities
$u_i(a, \theta)$. Note that each realized utilities has bounded range
$u_i^j(a)\in [0,1]$, by
Chernoff-Hoeffding bound and an application of union bound, we know
with probability at least $1 - \beta$, the average utilities satisfy
\[
\mbox{for all } i\in [n], a\in B, \qquad |\overline u_i(a) -
u_i(a)| \leq \sqrt{\frac{1}{2d}\ln\left(\frac{2|B|n}{\beta} \right)}
\]
Note that for $d \geq \frac{1}{\zeta^2} \ln\left(
\frac{2|B|n}{\beta}\right)$, then we will have $|\overline u_i(a) -
u_i(a)| < \zeta/ 2$ for all $i\in [n], a\in B$ with probability at
least $1-\beta$. We will condition on this event. This means
$\left\|\overline U - U(\theta)\right\|_\infty < \zeta$, so
$\widehat U = U(\theta)$ and therefore, $S' = S$.

\newpage
\section{Lower bound for stochastic utilities}
\label{sec:LB}

Let us argue that the $O(\log T)$ regret rate in Theorem~\ref{thm:model-noisy} is essentially optimal for Bayesian Exploration with stochastic utilities. We draw on the logarithmic lower bound in multi-armed bandits, a well-known result that dates back to \citet{Lai-Robbins-85}.%
\footnote{A more modern and lucid exposition of this result can be found in \citet[][Section 2.3]{Bubeck-survey12}.}

We start with a simpler formulation of our lower bound:

\begin{theorem}\label{thm:LB-regret-simple}
Consider Bayesian Exploration with stochastic utilities. Focus on a special case with only a single agent per round and only two actions. For any $C>0$ there is an absolute constant $\delta_0$ and a problem instance for which any iterative recommendation policy $\pi$ (BIC or not) suffers Bayesian regret
\begin{align}\label{eq:LB-regret-LB-simple}
\opt(\PiBIC[\delta])-\REW(\pi) \geq C\cdot \log t
    \qquad \text{for infinitely many rounds $t$ and any $\delta\in(0,\delta_0)$}.
 \end{align}
\end{theorem}

A stronger and more informative result (Theorem~\ref{thm:LB-regret}, which implies Theorem~\ref{thm:LB-regret-simple}) is formulated and proved later in this section, preceded by some background and discussion.

\subsection{Background and discussion}

\xhdr{Lower bounds in bandits.}
On a high level, a lower bound on regret in a multi-armed bandit problem asserts that a given family $\cF$ of problem instances is ``hard" for any algorithm. A precise statement of such result can take one of the following three forms found in the literature:
\begin{itemize}
\item[(i)] any algorithm incurs high expected regret on \emph{some} problem instance in $\cF$.
\item[(ii)] if a problem instance is drawn from some distribution over $\cF$, any algorithm incurs high expected regret.
\item[(iii)] any algorithm incurs high regret on \emph{every} problem instance in $\cF$, perhaps under some assumption on the algorithm.
\end{itemize}
Form (ii) is stronger than (i). Form (iii) without any assumption is stronger than (ii). The two most known lower bounds are the $\sqrt{T}$ bound from \citet{bandits-exp3} and the instance-dependent constant times $\log(T)$ bound from \citet{Lai-Robbins-85}; they are in the form (ii) and (iii), respectively. The Lai-Robbins lower bound has a weaker but more lucid corollary in the form (i).

In all bandit lower bounds that we are aware of, the construction of the family $\cF$ tends to be quite simple and intuitive, whereas the actual regret bound is obtained by a lengthy and intricate mathematical analysis based on KL-divergence. Thus, when proving a new lower regret bound it is very desirable to reduce it to an existing one rather than derive it from first principles.

\xhdr{Additional difficulties in Bayesian Exploration.}
Connecting Bayesian Exploration to the Lai-Robbins lower bound is subtle, for several reasons.

First, we'd like a lower bound for the same benchmark as the upper bound in Theorem~\ref{thm:model-noisy}, \ie for $\opt(\PiBIC[\delta])$, whereas the bandit lower bound is relative to the best action. While
in Bayesian Exploration some actions might not be explorable, we use the tools from \citet{ICexploration-ec15} to argue that all actions are explorable in the problem instances that we focus on. Moreover, the benchmark
    $\opt(\PiBIC[\delta])$
is about $\delta$-BIC policies, $\delta>0$, whereas the results from \citet{ICexploration-ec15} are only about $\delta=0$ as stated. We need to be careful with our assumptions so as to observe that the technique from \citet{ICexploration-ec15} yields in $\delta$-BIC policies, for any sufficiently small $\delta$. In the end, we conclude that
    $\opt(\PiBIC[\delta])$
equals the best-action benchmark.

Second, the Lai-Robbins lower bound includes a non-trivial assumption that, essentially, the algorithm's regret is not too high. So, the lower bound for Bayesian Exploration should either include a similar assumption explicitly, or derive it inside the proof (we accomplish the latter). Moreover, while a policy for Bayesian Exploration can be seen as a bandit algorithm, in Bayesian Exploration the algorithm inputs a prior, whereas in the Lai-Robbins lower bound it doesn't. In particular, the Lai-Robbins machinery immediately applies only when the prior is fixed. In order to resolve these difficulties and apply the Lai-Robbins lower bound after all, we focus on a carefully constructed instance of Bayesian Exploration, and formulate a non-standard version of Lai-Robbins whose proof is implicit in \citet[][Section 2.3]{Bubeck-survey12}.

\subsection{Formulation and proof}

We focus on a relatively simple special case of Bayesian Exploration for which the techniques from  \citet{ICexploration-ec15} and \citet{Lai-Robbins-85} can be usefully applied.

\begin{definition}
Consider a special case of Bayesian Exploration with only a single agent per round ($n=1$), and only two actions, $a_1$ and $a_2$. We posit that principal's reward equals agent's utility, and that the (realized) agents' utilities are Bernoulli random variables. The prior is independent:
\[ \Pr\sbr{ u(a_1) = \nu_1 \text{ and } u(a_2)=\nu_2} =
        \Pr[u(a_1)=\nu_1] \cdot \Pr[u(a_2)=\nu_2] ,\]
where $u(a)$ denotes the expected utility of action $a$.

Further, expected utilities can only take exactly three distinct values:
\begin{align*}
    u(a) \in \{\nu_1,\nu_2,\nu_3\} \subset (0,1)
    \qquad\text{for each action $a$},
\end{align*}
and the prior assigns a positive probability to each action-value pair:
    $\Pr[u(a_i)=\nu_j]>0$.
The two actions have distinct prior mean rewards:
    $\E[u(a_1)]\neq \E[ u(a_2)]$.

Such problem instances are called {\bf\em two-by-three instances} with support
    $\{\nu_1,\nu_2,\nu_3 \}$.
\end{definition}

\begin{lemma}\label{lm:LB-benchmarks}
For any two-by-three instance of Bayesian Exploration, there exists $\delta_0>0$ such that
    \[ B_\cU = \opt(\PiBIC[\delta])
        \qquad \text{for any $\delta\in (0,\delta_0)$}. \]
\end{lemma}

\begin{proof}
Fix a two-by-three problem instance, let $\cU$ denote its utility structure. Since the two actions have distinct prior mean rewards, without loss of generality assume that $\E[u(a_1)]> \E[u(a_2)]$.

It is easy to see that this problem instance satisfies the following property. Let $S^k$ be the the random variable representing the first $k$ rewards received from action $a_1$.
Let
\begin{align*}
X^k =  \E\sbr{\mu_2-\mu_{1} | S^k}.
\end{align*}
There exist constants $k_\cU, \tau_\cU, \rho_\cU>0$, possibly depending on $\cU$,
\begin{align*}
\Pr\sbr{ X_{i}^k> \tau_\cU} \geq \rho_{\cU} \qquad\forall k>k_\cU.
\end{align*}
\citet[][Section 4]{ICexploration-ec15} proved that with this property, both actions are eventually explorable. Moreover, they provided a BIC policy which explores both actions at least $k$ times, for any given $k$, and terminates at some round $T=T_{k,\cU}$ determined by the $k$ and the utility structure. Let us extend this policy by one more round, so that in the last round it recommends an arm with a highest posterior mean reward. It is easy to verify that this extended policy is in fact $\delta$-BIC for any $\delta\in (0,\delta_\cU)$ and $k>k'_\cU$, for some constants $\delta_\cU$ and $k'_\cU$ that are determined by the utility structure. Let $f(k)$ be the Bayesian-expected reward of this policy in the last round, as a function of parameter $k$. It is easy to see that
    $f(k) \to B_\cU$
as $k\to\infty$.
\end{proof}

We consider iterative recommendation policies $\pi$ that run indefinitely, without a built-in time horizon. In light of Lemma~\ref{lm:LB-benchmarks}, we focus on regret with respect to the best-action benchmark. For technical convenience, we define it here for each round $t$ and \emph{before} taking the expectation over the prior:
    \[ R(t) := t\cdot \sbr{ \max_{a\in\cA} u(a,\trueState)}
        - \sum_{s=1}^t u(\pi^s,\trueState). \]

\begin{theorem}\label{thm:LB-regret}
Fix two distinct numbers
    $\nu_1,\nu_2\in [\tfrac14,\tfrac34]$.
There exists a number
    $\nu \in (0,1)\setminus \{\nu_1,\nu_2\} $
and an absolute constant $C_0$ with the following property. Consider any two-by-three instance of Bayesian Exploration with support
    $\{\nu,\nu_1,\nu_2\}$
and prior probabilities
\begin{align}\label{eq:LB-thm-probs}
    \Pr\sbr{u(a_i)=\nu_j}\in [\tfrac14,\tfrac34], \quad i,j=\{1,2\}.
\end{align}
On this problem instance, any iterative recommendation policy (BIC or not) suffers regret
\begin{align}\label{eq:LB-regret-LB}
\E[R(t)]\geq \frac{C_0\cdot \log t}{|\nu_1-\nu_2|}
    \qquad \text{for infinitely many rounds $t$}.
\end{align}
\end{theorem}

To prove Theorem~\ref{thm:LB-regret}, we make a connection to a basic model of multi-armed bandits in which there are only two actions $a_1,a_2$, the reward of each action $a$ is a Bernoulli random variable with a fixed but unknown expectation $\mu(a)$, and no Bayesian prior is available to the algorithm. Henceforth, we refer to this model as \emph{Basic Bandits}. An instance of Basic Bandits is specified by the mean rewards $\mu(a_1),\mu(a_2)$.
Any iterative recommendation policy for a two-by-three instance can be interpreted as an algorithm for Basic Bandits.

We use a non-standard version of the Lai-Robbins lower bound, stated below. It is implicit in the proof of the standard version, as presented in \citet[][Section 2.3]{Bubeck-survey12}.

\begin{theorem}[\citet{Lai-Robbins-85,Bubeck-survey12}]\label{thm:LB-bandit-LB}
Fix two distinct numbers
    $\nu_1,\nu_2\in [\tfrac14,\tfrac34]$.
There exists a number
    $\nu \in (0,1)\setminus \{\nu_1,\nu_2\} $
with the following property. Consider any algorithm for Basic Bandits which achieves a ``not-too-high" regret for each problem instance with mean rewards in $\{\nu,\nu_1,\nu_2\}$. Specifically, assume this algorithm achieves regret
\begin{align}\label{eq:thm-standard-LB-assn}
\E[R(t)]\leq O(C_{\cI,\alpha}\; t^\alpha)
\quad\text{for each round $t$ and each $\alpha>0$},
\end{align}
where the constant $C_{\cI,\alpha}$ is determined by the problem instance and the $\alpha$. Then for any problem instance with distinct mean rewards in
    $\{\nu_1,\nu_2\}$
this algorithm suffers regret
\begin{align}\label{eq:thm-standard-LB-LB}
\E[R(t)]\geq \frac{C\cdot \log t}{|\nu_1-\nu_2|}
    \qquad \text{for any time $t>t_0$},
 \end{align}
where $C$ is an absolute constant and $t_0$ can depend on the problem instance.
\end{theorem}

\begin{proof}[Proof of Theorem~\ref{thm:LB-regret}]
Fix an iterative recommendation policy $\pi$, BIC or not. Pick $\nu$ as in Theorem~\ref{thm:LB-bandit-LB}. Fix a two-by-three problem instance $\cI$ with support
    $\{\nu,\nu_1,\nu_2\}$
and prior probabilities as in \eqref{eq:LB-thm-probs}.  For the sake of contradiction, suppose policy $\pi$ does not satisfy the lower bound \eqref{eq:LB-regret-LB} for this problem instance. Then, letting
    $c_0 = 1/|\nu_1-\nu_2|$,
for some time $t_0$ it holds that
\begin{align*}
\E[R(t)]\leq  O(c_0\,\log t)
    \qquad \text{for any time $t>t_0$}
\end{align*}
(where the $O()$ notation hides absolute constants). It follows that policy $\pi$ satisfies this  upper bound for any instance of Basic Bandits with mean rewards in $\{\nu,\nu_1,\nu_2\}$. Then by Theorem~\ref{thm:LB-bandit-LB} policy $\pi$ suffers the lower bound \eqref{eq:thm-standard-LB-LB}
for the instance of Basic Bandits with mean rewards
    $\mu(a_1) = \nu_1$ and $\mu(a_2) = \nu_2$.
This instance occurs with constant probability with respect to the prior in $\cI$, which implies the desired lower bound \eqref{eq:LB-regret-LB}.
\end{proof}

\end{document}